\documentclass[9pt,final,journal,twocolumn]{IEEEtran}

\IEEEoverridecommandlockouts
\usepackage{color}

\usepackage{amsthm,amsmath,amssymb}
\usepackage{mathtools}

\usepackage{verbatim}
\DeclareMathOperator{\Tr}{Tr}
\usepackage{cite}

\usepackage{url}
\usepackage{cases}
\usepackage{stmaryrd}

\usepackage[inline]{enumitem}

\usepackage{graphicx}

\ifCLASSOPTIONcompsoc
  \usepackage[caption=false,font=normalsize,labelfont=sf,textfont=sf]{subfig}
\else
  \usepackage[caption=false,font=footnotesize]{subfig}
\fi

\usepackage[switch]{lineno}
\usepackage[named]{algo}
\usepackage[noend]{algpseudocode}
\usepackage{algorithm}

\usepackage[flushleft]{threeparttable} 
\usepackage{bm}
\newtheorem{theorem}{Theorem}

\newtheorem{remark}{Remark}

\usepackage{caption}
\begin{document}

\title{Fundamental Limits for Near-Field Sensing --- Part I: Narrow-Band Systems}

\author{Tong Wei,~\IEEEmembership{Member,~IEEE}, Kumar Vijay Mishra, \IEEEmembership{Senior Member,~IEEE}, 
Bhavani Shankar M.R., \IEEEmembership{Senior Member,~IEEE}, Björn Ottersten,~\IEEEmembership{Fellow,~IEEE}


\thanks{The authors are with the Interdisciplinary Centre for Security, Reliability and Trust (SnT), University of Luxembourg, Luxembourg City L-1855, Luxembourg. E-mail: \{tong.wei@, kumar-mishra@ext.,bhavani.shankar@, bjorn.ottersten@\}uni.lu.}
}

\maketitle 
\vspace{-2.0em}
\vspace{-4cm}
\begin{abstract}
Extremely large-scale antenna arrays (ELAAs) envisioned for 6G enable high-resolution sensing. However, the ELAAs worked in extremely high frequency will push operation into the near-field region, where spherical wavefronts invalidate classical far-field models and alter fundamental estimation limits. The purpose of this and the companion paper (Part II) is to develop the
theory of fundamental limits for near-field sensing systems in detail.
In this paper (Part I), we develop a unified narrow-band near-field signal model for joint parameter sensing of moving targets using the ELAAs. Leveraging the Slepian–Bangs formulation, we derive closed-form Cramér–Rao bounds (CRBs) for joint estimation of target position, velocity, and radar cross-section (RCS) under the slow-time sampling model. To obtain interpretable insights, we further establish explicit far-field and near-field approximations that reveal how the bounds scale with array aperture, target range, carrier wavelength, and coherent integration length. The resulting expressions expose the roles of self-information terms and their cross terms, clarifying when Fresnel corrections become non-negligible and providing beamformer and algorithm design guidelines for near-field sensing with ELAAs. Simulation results validate the derived CRBs and their far-field and near-field approximations, demonstrating accurate agreement with the analytical scaling laws across representative array sizes and target ranges. 
\end{abstract}


\IEEEpeerreviewmaketitle

\section{Introduction}
The emergence of extremely large-scale antenna arrays (ELAAs) with synthetic-aperture architectures in extremely high frequencies (EHF) is driving a paradigm shift from conventional far-field operation to near-field communications to enable sixth-generation (6G) wireless systems \cite{elbir2024terahertz,elbir2024curse,hodge2020intelligent}. 
These systems aim to provide ultra-reliable, low-latency communication, massive connectivity, ultra-high data rates, and precise positioning \cite{wachowiak2026sizing,wang2023on}. These advancements also pave the way for integrating sensing capabilities into 6G networks, addressing critical challenges such as spectrum congestion, hardware costs, and power consumption \cite{mishra2024signal,mishra2019toward,wei2024RIS,vargas2023dual}. However, the combination of ELAA and EHF extends the Rayleigh distance, necessitating operation in the near-field region \cite{balanis2011modern,zohair2021near,bevilacqua2025reflectivity,rajwan2025near,jiang2024near}. This shift underscores the need for accurate near-field signal models and sophisticated processing algorithms to fully harness the potential of these technologies \cite{elbir2025near}.
 
Near-field sensing leverages electrically large apertures operating in the Fresnel region, where the received field is no longer well approximated by a plane wave model and the array response depends jointly on range and angle through spherical-wave propagation \cite{elbir2025near}. This inherent range–angle dependence enables near-field beam focusing and provides richer spatial signatures than conventional far-field processing \cite{wei2025wideband,wei2025crb}, which can further enhance resolvability for localization, tracking, and imaging with large arrays and synthetic apertures \cite{wu2021modified,jiang2024near,he2025grating,eamaz2023near,guerra2021near}. Consequently, near-field sensing has recently become an important building block for 6G applications including high-precision positioning, mapping, short-range imaging for robotics and autonomous systems, and environment-aware perception supported by programmable or holographic apertures \cite{zohair2021near,bevilacqua2025reflectivity,rajwan2025near,jiang2024near,wang2025perofrmance}. More recently, near-field integrated sensing and communication (ISAC) is attracting significant interest as a spectrum, energy, and hardware-efficient 6G paradigm, where the shared aperture, RF front-end, and waveform are used to support both data transmission and sensing \cite{wang2023near,zhao2024modeling}.

To quantify the achievable estimation accuracy, the Cramér–Rao bound (CRB) is widely adopted as a fundamental metric to guide waveform design, beamforming, and sensing-oriented resource allocation for parameter estimation and tracking \cite{esmaeilbeig2022cramer,dokhanchi2019mmWave,lv2022co,wang2024crb2}. Nevertheless, most of the classical CRB literature is built on far-field assumptions, where propagation is modeled by planar wavefronts and the path loss is often treated as constant over the signal bandwidth and antenna element \cite{stoica2005spectral,dogandzic2001crb}. Under these conditions, the range, angle, and Doppler are frequently weakly coupled or even separable, yielding comparatively simple Fisher information structures and tractable closed-form bounds \cite{dogandzic2001crb}. However, once the operating region enters the near field, these far-field simplifications no longer apply. In particular, spherical-wave propagation induces an explicit range dependence in the array response and tightens the coupling between spatial parameters and motion over slow time. Motivated by this, Part I develops a dedicated CRB for narrow-band near-field sensing that captures these effects and quantifies the resulting performance limits.

\vspace{-0.3cm}
\subsection{Related Works}
\vspace{-0.1cm}
Estimation-theoretic bounds, most notably the CRB, provide a rigorous benchmark for the best achievable accuracy under a specified observation model and noise statistics and are therefore widely used in synthetic-aperture design, direction of arrival (DOA), and channel estimation \cite{stoica1990performance,dong2002optimal,wei2025crb}. In the narrow-band near-field setting, Gazzah et.al. derived the matrix-free conditional CRBs for joint angle and range estimation using the exact delay model for arbitrarily spaced linear arrays, and further showed how the bound supports CRB-driven sensor placement beyond classical ULAs \cite{gazzah2014crb}.


Complementary work has investigated near-field localization bounds using more physically grounded electromagnetic field formulations~\cite{torres2021crb} and has studied near-field tracking limits, i.e., position and velocity, using phase profiles across large arrays~\cite{guerra2021near}. Beyond co-located array setups, near-field localization in bistatic MIMO radar with arbitrary array geometries was investigated in \cite{khamidullina2021conditional}, where both deterministic and stochastic CRBs were derived under exact near-field delay and amplitude profiles, and representative estimators were benchmarked against these limits.
Motivated by the requirement of future 6G systems with ELAAs, closed-form near-field CRB characterizations for angle and range were developed in \cite{wang2024crb}, revealing key scaling behaviors with aperture and antenna number, and clarifying the accuracy limits of common Taylor-type approximations in CRB analysis. Complementarily, \cite{meng2024crb} proposed a generic modular ELAA architecture and obtained tractable CRBs under a hybrid spherical and planar wave model (near-field for the full array but far-field within each subarray), highlighting the impact of subarray diversity and edge-distributed deployments.
Performance bounds for near-field localization with widely spaced multi-subarray mmWave and THz MIMO architectures were reported in \cite{yang2024performance}, further emphasizing the interplay between array topology and fundamental limits.

Different from the previous works that primarily focus on angle and range estimation for either static or single target setup, Part~I of this paper considers a narrow-band near-field sensing and develops unified CRB characterizations for \emph{RCS, velocity, and location} with multiple targets, together with the comparison of near-field and far-field approximations, and scaling insights tailored to near-field systems. Based on these foundations, Part~II of this paper \cite{wei2025PartII} will extend those fundamental limits for near-field sensing in wide-band systems.

\vspace{-0.3cm}

\subsection{Contributions}
\vspace{-0.1cm}
In this paper (\textbf{Part I}), we establish fundamental limits for near-field sensing under narrow-band signaling and provides closed-form insights suitable for system design. The contributions of this paper over prior-art are listed below:
\begin{enumerate}
   \item  \textbf{Unified Narrow-band Near-field  Sensing Model:} We establish a unified narrow-band near-field signal model for synthetic-aperture sensing under the general monostatic or bistatic geometry, explicitly capturing spherical-wave propagation, element-dependent path loss, and slow-time Doppler evolution for moving targets.
   \item  \textbf{General CRBs for Parameter Estimation:} Based on the proposed model, we derive the Fisher information matrix (FIM) and conditional CRBs for key target parameters, including RCS, velocity, and location, while retaining the near-field coupling effects between position- and motion-related parameters. 
   \item  \textbf{Asymptotic Closed-form Expressions:}  To enable tractable analysis and reveal design insights, we develop the near-field approximations and obtain asymptotic closed-form CRB expressions that characterize the scaling laws with respect to range, antenna number, SNR, and snapshot number, thereby clarifying when near-field approximation become valid.
   \item  \textbf{Extensive Performance Evaluation:} 
   We provide comprehensive numerical experiments to validate the theoretical CRBs and illustrate the evolution of performance from classical far-field approximation to proposed near-field approximation, serving as a benchmark for future algorithm design and measurement verification.
\end{enumerate}

\vspace{-0.3cm}
\subsection{Organization}
\vspace{-0.1cm}
The remainder of Part~I is organized as follows. Section~II introduces the system geometry and the narrow-band near-field signal model. Section~III derives the Fisher information matrix, presents conditional narrow-band CRBs for RCS, velocity, and 2D location, and then provides asymptotic CRB expressions under near-field and far-field approximations and discusses the resulting scaling laws and design insights. Section~IV reports numerical experiments to validate the analysis and quantify the accuracy of the proposed approximations. Section~V concludes Part~I and briefly outlines the extension to wideband systems in Part~II.

\begin{figure}[t]
\centering{\includegraphics[width=1\columnwidth]{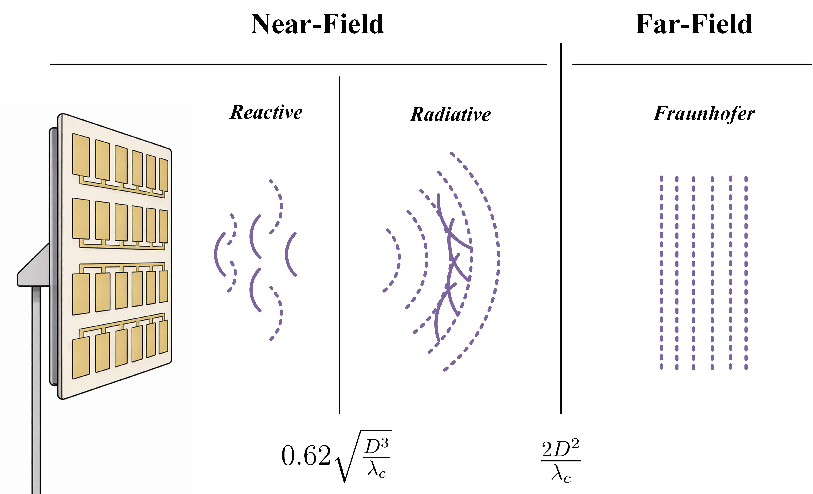}}
\caption{Comparison of reactive near-field, Fresnel near-field, and Fraunhofer far-field regions of an aperture.
\label{dig}}
\end{figure}

\vspace{-0.3cm}
\subsection{Notations}
\vspace{-0.1cm}
Throughout this paper, lowercase and uppercase boldface letters represent vectors and matrices, respectively. The classical operations and notations include  Transposition, $(\cdot)^T$, complex conjugation, $(\cdot)^{\ast}$,  Hermitian $(\cdot)^H$, matrix vectorization, $\mathrm{vec}(\cdot)$, a diagonal matrix, $\mathrm{diag}(\cdot)$,  and a block diagonal matrix, $\mathrm{bdiag}(\cdot)$. 
%
\section{Narrow-band Signal Model}
Consider a narrow-band near-field sensing system operating at carrier frequency $f_c$ with wavelength $\lambda_c=c/f_c$. Both transmitter (Tx) and  receiver (Rx) employ ULAs with $N_t$ and $N_r$ antenna elements deployed in a 2D Cartesian plane with the inter-element spacings being, $d_t$ and $d_r$, respectively. The coordinates of the $n_t$-th Tx and $n_r$-th Rx elements are denoted by ${\mathbb I}_{n_t}=(x_{n_t},y_{n_r})$ and ${\mathbb I}_{n_r}=(x_{n_r},y_{n_r})$, where $n_t\in\{1,\ldots,N_t\}$ and $n_r\in\{1,\ldots,N_r\}$. The radiative near-field sensing region contains $Q$ sensing targets, including the radar targets, the communication users, and static scatterers \cite{hu2024joint}. Fig.\ref{dig} illustrates the wavefront evolution across regions: reactive near-field, radiative (Fresnel) near-field, and far-field (Fraunhofer), separated by the approximate boundaries $0.62\sqrt{D^{3}/\lambda_c}$ and $2D^{2}/\lambda_c$, respectively.
These targets are indexed by $q\in\{1,\ldots,Q\}$, whose 2D location and 2D velocity are ${\mathbb I}_q=(x_q,y_q)$ and ${\bm v}^{(q)}=[v_x^{(q)},v_y^{(q)}]^T$, respectively\footnote{In ISAC systems, scatterers are typically assumed to be static compared to the moving communication users and radar targets\cite{mishra2024signal,mishra2019toward}.}.
%
\subsection{Narrow-band Transmit Signal Model}
Assume the Tx emits $M$ narrow-band snapshots within a coherent processing interval (CPI) of length $T_{\mathrm{CPI}}=MT_{\mathrm{sym}}$, where $T_{\mathrm{sym}}$ is the snapshot repetition interval \cite{xi2020gridless}. Let ${\bf x}(m)\in\mathbb{C}^{N_t\times 1}$ denote the known transmit symbol vector at the slow-time index $m$, which is given by
\begin{align}
{\bf x}(m)=[x_{1}(m),\cdots,x_{N_t}(m)]^T.
\end{align}
Herein, we consider the commonly used isotropic transmission strategy, 
and then we have 
\begin{equation}\label{iso_trans}
\mathbb{E}\left\{\frac{1}{M}\sum_{m=0}^{M-1}{\bf x}(m){\bf x}^H(m)\right\}=\mathcal{P}{\bf I}_{N_t}.
\end{equation}
where $\mathcal{P}$ denotes the transmit power. 
The continuous-time baseband signal representation is given by
\begin{align}\label{eq:NB_tx_bb}
{\bf x}(t)=\sum_{m=1}^{M}{\bf x}(m)\mathrm{rect}\left(\frac{t-mT_{\mathrm{sym}}}{T_{\mathrm{sym}}}\right),
\end{align}
where $\mathrm{rect}(\cdot)$ denotes the rectangular function, i.e., $\mathrm{rect}(x)=1$ for $x\in[0,1)$ and $0$ otherwise.
Then, the transmitted pass-band signal of Tx is then given
\begin{align}\label{eq:NB_tx_pb}
\widetilde{\bf x}(t)=\Re\left\{{\bf x}(t)e^{j2\pi f_ct}\right\}.
\end{align}
\subsection{Assumptions for Narrow-band Near-Field Modeling}
To obtain a tractable narrow-band near-field model, we adopt the following standard assumptions.
\begin{description}
\item[A1] ``Narrow-band Propagation:'' The effective signal bandwidth $B$ satisfies $B\ll f_c$, such that a propagation delay can be modeled as a phase rotation at $f_c$, i.e., $\widetilde{x}_{n_t}(t-\tau)\approx \widetilde{x}_{n_t}(t)e^{-j2\pi f_c\tau}$ after demodulation and matched filtering.
\item[A2] ``Snapshot Invariant RCS:'' RCS, i.e., $\alpha_q$, is constant over the CPI and does not vary with $m\in\{1,\cdots,M\}$.
\item[A3] ``Negligible Multi-Bounce Reflections:'' Higher-order reflections are absorbed into the effective noise \cite{xu2023bandwidth}.
\item[A4] \emph{Constant-velocity and small displacement:}
The target velocity ${\bm v}^{(q)}$ is constant over the CPI and the displacement is negligible relative to the bistatic ranges, i.e.,
$\|{\bm v}^{(q)}\|T_{\mathrm{CPI}} \ll \mathrm{min} \left(r_{n_t}^{(q)}, r_{n_r}^{(q)}\right)$, where $r_{n_t}^{(q)}, r_{n_r}^{(q)}$
denote the bistatic ranges (kindly refer \eqref{NF_dis} below). Hence, the distance-dependent pathloss terms are treated as snapshot-invariant, while Doppler is captured through the phase evolution.
\end{description}
Based on these assumptions, we will derive the corresponding narrow-band near-field receive model by incorporating the propagation delays and Doppler-induced phase evolution across snapshots in the sequel.
\subsection{Near-Field Receive Signal Model}
Let the  bistatic distances from $q$-th target to the $n_t$-th Tx element and $n_r$-th Rx element be 
\begin{align} \label{NF_dis}
r_{n_t}^{(q)}=\|{\mathbb I}_q-{\mathbb I}_{n_t}\|,\quad
r_{n_r}^{(q)}=\|{\mathbb I}_q-{\mathbb I}_{n_r}\|.
\end{align}
Then, the corresponding bistatic propagation delay is
\begin{align}\label{eq:NB_delay}
\tau_{n_tn_r}^{(q)}=\frac{r_{n_t}^{(q)}+r_{n_r}^{(q)}}{c}.
\end{align}
Based on \eqref{NF_dis} and \eqref{eq:NB_delay}, the pass-band signal received at the $n_r$-th Rx antenna is modeled as
\begin{small}
\begin{flalign}\label{eq:NB_rx_pb}
&\widetilde y_{n_r}(t)
\!=\!\sum_{q=1}^{Q}\alpha_qg_{n_r}^{(q)}
\sum_{n_t=1}^{N_t}g_{n_t}^{(q)}
\widetilde x_{n_t}(t\!-\!\tau_{n_tn_r}^{(q)})
e^{j2\pi f_{n_tn_r}^{(q)}t}
\!+\!\widetilde n_{n_r}(t),
\end{flalign}
\end{small}
\noindent\hspace{-1.0em} where $\widetilde n_{n_r}(t)$ denotes the receiver noise, $f_{n_tn_r}^{(q)}$ denotes the Doppler shift of the path between the $n_t$-th Tx element and the $n_r$-th Rx element; this can be modeled as \cite{sun2024widely}
\begin{align}\label{eq:NB_doppler}
f_{n_tn_r}^{(q)}
=\frac{f_c}{c}\left(
\frac{\langle{\bm v}^{(q)},{\mathbb I}_q-{\mathbb I}_{n_t}\rangle}{\|{\mathbb I}_q-{\mathbb I}_{n_t}\|}
+
\frac{\langle{\bm v}^{(q)},{\mathbb I}_q-{\mathbb I}_{n_r}\rangle}{\|{\mathbb I}_q-{\mathbb I}_{n_r}\|}
\right).
\end{align}
$g_{n_t}^{(q)}$ and $g_{n_r}^{(q)}$ denote the distance-dependent pathloss modeled as 
\begin{align}\label{eq:NB_pathloss}
g_{n_t}^{(q)}=\frac{\sqrt{c_{n_t}}\lambda_c}{4\pi \|{\mathbb I}_q-{\mathbb I}_{n_t}\|},\quad
g_{n_r}^{(q)}=\frac{\sqrt{c_{n_r}}\lambda_c}{4\pi\|{\mathbb I}_q-{\mathbb I}_{n_r}\|}.
\end{align}
where ${c_{\ast}},\ast\in\{n_r,n_t\}$ denotes the radiation profile of the corresponding antenna element, i.e., $c_{\ast}=1$ for full-digital or hybrid digital-analog arrays \cite{elbir2024terahertz,elbir2023twenty,elbir2021terahertz,elbir2023nba} and $0<c_{\ast}<1$ for hybrid digital-holographic array \cite{xu2024near}. Hereafter, for simplicity, we only consider the full-digital case. 

After I/Q demodulation and matched filtering, we can sample once per snapshot at slow-time index $m$. Based on \textbf{A3}, the Doppler-induced phase is approximately constant within snapshot $m$ and can be evaluated at $t=mT_{\mathrm{sym}}$ \cite{richards2010principles}. The resulting discrete-time measurement at the $n_r$-th Rx element is
\begin{align}\label{eq:NB_rx_discrete_scalar}
y_{n_r}(m)
=
&\sum_{q=1}^{Q}\alpha_q
\sum_{n_t=1}^{N_t}
g_{n_t}^{(q)}g_{n_r}^{(q)}
x_{n_t}(m) \nonumber\\
&\cdot e^{-j2\pi f_c\tau_{n_tn_r}^{(q)}}
e^{j2\pi f_{n_tn_r}^{(q)}mT_{\mathrm{sym}}}
+n_{n_r}(m),
\end{align}
where $n_{n_r}(m)\sim\mathcal{CN}(0,\sigma_r^2)$.
Stacking the receiving signal in \eqref{eq:NB_rx_discrete_scalar}, we have 
\begin{align}\label{eq:NB_rx_vector}
{\bf y}(m)
=\sum_{q=1}^{Q}\alpha_q{\bf a}_{\mathrm{R}}(m,q){\bf a}_{\mathrm{T}}^{T}(m,q){\bf x}(m)+{\bf n}(m)
\end{align}
where ${\bf y}(m)\in\mathbb{C}^{N_r\times 1}$, ${\bf n}(m)\sim\mathcal{CN}({\bf 0},\sigma_r^2{\bf I}_{N_r}),$
and the narrow-band near-field steering vectors are
\begin{subequations} \label{eq:NB_steering}
\begin{align}
{\bf a}_{\mathrm{T}}(m,q)
&=
\Big[g_{1}^{(q)}e^{\varphi_{1}(m,q)},\ldots,g_{N_t}^{(q)}e^{\varphi_{N_t}(m,q)}\Big]^T,  \label{eq:NB_steering_Tx}\\
{\bf a}_{\mathrm{R}}(m,q)
&=
\Big[g_{1}^{(q)}e^{\varphi_{1}(m,q)},\ldots,g_{N_r}^{(q)}e^{\varphi_{N_r}(m,q)}\Big]^T, \label{eq:NB_steering_Rx}
\end{align}
\end{subequations}
with
\begin{small}
\begin{subequations} \label{eq:NB_phase_TxRx}
\begin{flalign}
\varphi_{n_t}(m,q)
&\!=\!\frac{j2\pi f_c}{c}\left(
\frac{\langle{\bm v}^{(q)},{\mathbb I}_q\!-\!{\mathbb I}_{n_t}\rangle}{\|{\mathbb I}_q\!-\!{\mathbb I}_{n_t}\|}mT_{\mathrm{sym}}
\!-\!\|{\mathbb I}_q\!-\!{\mathbb I}_{n_t}\|\!
\right),\\
\varphi_{n_r}(m,q)
&\!=\!\frac{j2\pi f_c}{c}\left(
\frac{\langle{\bm v}^{(q)},{\mathbb I}_q\!-\!{\mathbb I}_{n_r}\rangle}{\|{\mathbb I}_q\!-\!{\mathbb I}_{n_r}\|}mT_{\mathrm{sym}}
\!-\!\|{\mathbb I}_q\!-\!{\mathbb I}_{n_r}\|\!
\right).
\end{flalign}
\end{subequations}
\end{small}
Defining ${\bf B}=\mathrm{diag}(\alpha_1,\ldots,\alpha_Q)$, ${\bf A}_{\mathrm{Tx},m}=[{\bf a}_{\mathrm{T}}(m,1),\ldots,{\bf a}_{\mathrm{T}}(m,Q)]$, and ${\bf A}_{\mathrm{Rx},m}=[{\bf a}_{\mathrm{R}}(m,1),\ldots,{\bf a}_{\mathrm{R}}(m,Q)]$, \eqref{eq:NB_rx_vector} can be compactly rewritten as
\begin{align}\label{eq:NB_compact}
{\bf y}(m)
=
{\bf A}_{\mathrm{Rx},m}{\bf B}{\bf A}_{\mathrm{Tx},m}^{T}{\bf x}(m)
+{\bf n}(m).
\end{align}
We can further stack all $M$ snapshots yields
\begin{align}\label{eq:NB_stack}
{\bf y}=[{\bf y}(1)^T,\ldots,{\bf y}(m)^T]^T
=
{\bf A}{\bf x}+{\bf n},
\end{align}
where ${\bf x}=[{\bf x}(1)^T,\ldots,{\bf x}(m)^T]^T$, ${\bf n}=[{\bf n}(1)^T,\ldots,{\bf n}(m)^T]^T$, and ${\bf A}=\mathrm{bdiag}({\bf A}_1,\ldots,{\bf A}_M)$ with ${\bf A}_m={\bf A}_{\mathrm{Rx},m}{\bf B}{\bf A}_{\mathrm{Tx},m}^{T}$. 
Meanwhile, we collect all unknown real-valued parameters as
\begin{align}\label{eq:NB_param}
{\bm\theta}=[{\bm x},{\bm y},{\bm v}_x,{\bm v}_y,\bm\alpha_{R},\bm\alpha_{I}]^T\in \mathbb{R}^{6Q\times1},
\end{align}
where ${\bm x}=[x_1,\ldots,x_Q]$, ${\bm y}=[y_1,\ldots,y_Q]$, ${\bm v}_x=[v_x^{(1)},\ldots,v_x^{(Q)}]$, ${\bm v}_y=[v_y^{(1)},\ldots,v_y^{(Q)}]$, and $\bm\alpha_{R}$ and $\bm\alpha_{I}$ collect the real and imaginary parts of $\{\alpha_q\}_{q=1}^Q$.

\begin{remark}
In the narrow-band near-field sensing system, the effective signal bandwidth $B$ is assumed to be much smaller than the carrier frequency $f_c$, so the propagation delay associated with each target can be absorbed into a carrier-phase rotation at $f_c$. Consequently, the observation does not explicitly resolve the delay dimension, and the extractable information is primarily carried by: (1) the near-field spherical-wavefront phase variation across the array aperture; (2) the slow-time phase evolution induced by Doppler. This renders narrow-band sensing largely aperture-limited, and targets with similar bistatic ranges may be harder to separate compared to the wideband case which will be discussed in \textbf{Part II}.
\end{remark}
\section{Narrow-band CRB Analysis}
To evaluate the parameter estimation performance, the CRB has been widely utilized as an estimation lower bound~\cite{stoica1990performance,liang2021CramerRao}. 
In this section, we derive the narrow-band near-field CRBs for joint estimation of position, velocity, and RCS of moving targets. 
Let us first rewrite the stacked receive model in real form as
$\bar{\bf y}=[\Re\{({\bf y})\}^T;\Im\{({\bf y})\}^T]^T$,
$\bar{\bm\mu}=[\Re\{({\bm\mu})\}^T;\Im\{({\bm\mu})\}^T]^T$, and
$\bar{\bf C}=\frac{\sigma_r^2}{2}{\bf I}_{2N_rM}$,
where ${\bm\mu}$ and ${\bf C}$ denote the mean and covariance of ${\bf y}$, respectively, i.e., ${\bf y}\sim({\bm\mu},{\bf C})$. 
Under the assumption \textbf{A2} and additive white Gaussian noise (AWGN) assumption\footnote{Note that \eqref{dist._rx_NB} holds due to the AWGN assumption. If this assumption does not hold, i.e., the noise is not circular or white, additional parameters are required to characterize ${\bf C}$, which is beyond the scope of this paper and is left for future work.}, we have 
\begin{subequations} \label{dist._rx_NB}
\begin{align}
 {\bm\mu} &= {\bf A}{\bf x}, \label{statistic1_NB}\\
 {\bf C}  &= \mathbb{E}\{({\bf y}-{\bm\mu})({\bf y}-{\bm\mu})^H\}=\sigma_r^2{\bf I}_{N_r M}. \label{statistic2_NB}
\end{align}
\end{subequations}
This likelihood directly enables the FIM and CRB derivations in the sequel.
To evaluate the CRB, we can write the likelihood function as~\cite{stoica2005spectral}
\begin{small}
\begin{align}\label{likelihood_NB}
 p(\bar{\bf y},{\bm\theta}) = \frac{1}
 {\sqrt{2\pi^{2{N}_rM} |\bar{\bf C}|}}
 \exp\!\left(-\frac{(\bar{\bf y}-\bar{\bm\mu})^T\bar{\bf C}^{-1}(\bar{\bf y}-\bar{\bm\mu})}{2}\right).
\end{align}
\end{small}
\noindent\hspace{-0.5em} According to \eqref{likelihood_NB}, the log-likelihood function is
\begin{flalign}\label{likelihood_log_NB}
\ln p 
 \!=\! -\frac{1}{2}\Big(\ln|\bar{\bf C}|+
 (\bar{\bf y}-\bar{\bm\mu})^T\bar{\bf C}^{-1}(\bar{\bf y}-\bar{\bm\mu})\Big),
\end{flalign}
where the constant term $MN_r\ln(2\pi)$ is omitted hereafter. 
Based on the above, the near-field FIM is given by~\cite{stoica2005spectral}
\begin{align}\label{FIM_gen_NB}
 {\bf F} = -\mathbb{E}\left\{ \frac{\partial^2\ln p}{\partial{\bm\theta}\,\partial{\bm\theta}^T}\right\},
\end{align}
where $\mathbb{E}\{\cdot\}$ denotes the expectation operator.
Based on the Slepian--Bangs formula, the $(i,j)$-th entry of the FIM can be written (see, e.g., (B.3.25) in~\cite{stoica2005spectral}) as
\begin{small}
\begin{flalign}\label{FIM_spe_NB}
&\hspace*{-0.5em}{\bf F}(i,j) \!=\!\Tr\!\left({\bf C}^{-1}\frac{\partial{\bf C}}{\partial{\theta}_i}{\bf C}^{-1}\frac{\partial{\bf C}}{\partial{\theta}_j}\right)
\!+\!2\Re\!\left[\left(\frac{\partial{\bm\mu}}{\partial{\theta}_i}\right)^{\!H}{\bf C}^{-1}\frac{\partial{\bm\mu}}{\partial{\theta}_j}\right],
\end{flalign}
\end{small}
\noindent\hspace{-0.5em} in which ${\bf C}=\sigma_r^2{\bf I}_{N_rM}$ does not depend on ${\bm\theta}$, hence $\frac{\partial{\bf C}}{\partial{\theta}_i}={\bf 0}$ and the first trace term vanishes. 
Therefore, \eqref{FIM_spe_NB} can be further simplified to
\begin{align}\label{FIM_spe_simplified_NB}
{\bf F}(i,j)=\frac{2}{\sigma_r^2}\Re\!\left[\left(\frac{\partial{\bm\mu}}{\partial{\theta}_i}\right)^{\!H}\frac{\partial{\bm\mu}}{\partial{\theta}_j}\right].
\end{align}
It follows that the key step in FIM evaluation is computing $\frac{\partial{\bm\mu}}{\partial{\theta}_i}$ for each parameter.
From ${\bm\mu}={\bf A}{\bf x}$ with ${\bf A}=\mathrm{bdiag}[{\bf A}_1,\ldots,{\bf A}_M]$, we have
\begin{align} \label{derivatives_mu_NB}
\frac{\partial{\bm\mu}}{\partial{\theta}_i}
&=
\left[\begin{matrix}
\frac{\partial({\bf A}_1{\bf x}(1))}{\partial{\theta}_i}\\
\vdots\\
\frac{\partial({\bf A}_M{\bf x}_M)}{\partial{\theta}_i}
\end{matrix}\right], \quad
\frac{\partial{\bm\mu}}{\partial{\theta}_j}
=
\left[\begin{matrix}
\frac{\partial({\bf A}_1{\bf x}(1))}{\partial{\theta}_j}\\
\vdots\\
\frac{\partial({\bf A}_M{\bf x}_M)}{\partial{\theta}_j}
\end{matrix}\right]. 
\end{align}
Then, the full FIM can be partitioned according to the structure of ${\bm\theta}$ as
\begin{small}
\begin{align}\label{FIM_block_NB}
\hspace*{-1.0em}{\bf F}  
\!=\!&\left[\begin{matrix}  
 {\bf F}_{{\bf x}{\bf x}} \!&\! {\bf F}_{{\bf x}{\bf y}} \!&\! {\bf F}_{{\bf x}{\bf v}_x} \!&\! {\bf F}_{{\bf x}{\bf v}_y} \!&\! {\bf F}_{{\bf x}{\bm\alpha}_R} \!&\! {\bf F}_{{\bf x}{\bm\alpha}_I}  \\
 {\bf F}_{{\bf y}{\bf x}} \!&\! {\bf F}_{{\bf y}{\bf y}} \!&\! {\bf F}_{{\bf y}{\bf v}_x} \!&\! {\bf F}_{{\bf y}{\bf v}_y} \!&\! {\bf F}_{{\bf y}{\bm\alpha}_R} \!&\! {\bf F}_{{\bf y}{\bm\alpha}_I}\\
  {\bf F}_{{\bf v}_x{\bf x}} \!&\! {\bf F}_{{\bf v}_x{\bf y}} \!&\! {\bf F}_{{\bf v}_x{\bf v}_x} \!&\! {\bf F}_{{\bf v}_x{\bf v}_y} \!&\! {\bf F}_{{\bf v}_x{\bm\alpha}_R} \!&\! {\bf F}_{{\bf v}_x{\bm\alpha}_I}\\
  {\bf F}_{{\bf v}_y{\bf x}} \!&\! {\bf F}_{{\bf v}_y{\bf y}} \!&\! {\bf F}_{{\bf v}_y{\bf v}_x} \!&\! {\bf F}_{{\bf v}_y{\bf v}_y} \!&\! {\bf F}_{{\bf v}_y{\bm\alpha}_R} \!&\! {\bf F}_{{\bf v}_y{\bm\alpha}_I}\\
  {\bf F}_{{\bm\alpha}_R{\bf x}} \!&\! {\bf F}_{{\bm\alpha}_R{\bf y}} \!&\! {\bf F}_{{\bm\alpha}_R{\bf v}_x} \!&\! {\bf F}_{{\bm\alpha}_R{\bf v}_y} \!&\! {\bf F}_{{\bm\alpha}_R{\bm\alpha}_R} \!&\! {\bf F}_{{\bm\alpha}_R{\bm\alpha}_I}\\
  {\bf F}_{{\bm\alpha}_I{\bf x}} \!&\! {\bf F}_{{\bm\alpha}_I{\bf y}} \!&\! {\bf F}_{{\bm\alpha}_I{\bf v}_x} \!&\! {\bf F}_{{\bm\alpha}_I{\bf v}_y} \!&\! {\bf F}_{{\bm\alpha}_I{\bm\alpha}_R} \!&\! {\bf F}_{{\bm\alpha}_I{\bm\alpha}_I}\\
    \end{matrix}\right],
\end{align}
\end{small}
\noindent\hspace{-1em} where the diagonal blocks correspond to the information on position, velocity, and RCS parameters, and the off-diagonal blocks capture their mutual coupling. 
For an unbiased estimator, the CRB matrix is given by the inverse FIM, i.e., ${\mathrm {CRB}}={\bf F}^{-1}$. 
When focusing on a subset of parameters (e.g., position only), we use Schur-complement techniques~\cite{dokhanchi2019mmWave} to extract the corresponding sub-block of ${\bf F}^{-1}$, which yields the tightest CRB for that subset.

\vspace{-0.2cm}
\subsection{Conditional CRB for RCS Estimation}
\vspace{-0.1cm}
We first recall the complex-valued partial-derivative identities \cite{hjrungnes2011complex}
\begin{flalign}\label{FIM_Derivation_NB}
  &\frac{\partial(a+jb)e^{jc}}{\partial a}=e^{jc},~~
  \frac{\partial(a+jb)e^{jc}}{\partial b}=je^{jc}.
\end{flalign}
Then, the partial derivatives with respect to the real and imaginary parts of the $q$-th reflectivity coefficient are
\begin{subequations}\label{FIM_Derivation_RCS_NB}
\begin{align}
 \frac{\partial{\bf A}_{m}{\bf x}_{m}}{\partial\alpha_{R_q}}
 &= {\bf a}_{\mathrm{R}}(m,q){\bf a}_{\mathrm{T}}^{T}(m,q){\bf x}_{m}, \\
 \frac{\partial{\bf A}_{m}{\bf x}_{m}}{\partial\alpha_{I_q}}
 &= j\,{\bf a}_{\mathrm{R}}(m,q){\bf a}_{\mathrm{T}}^{T}(m,q){\bf x}_{m},
\end{align}
\end{subequations}
for $q=1,\ldots,Q$.
Using the Cauchy-Riemann relations \cite{hjrungnes2011complex}, we obtain
\begin{subequations}\label{FIM_RCS_NB}
\begin{align}
  \frac{\partial{\bf A}_{m}{\bf x}_{m}}{\partial{\bm\alpha}_{R}}
  &= [\breve{\bf A}_{m}(1){\bf x}_{m},\ldots,\breve{\bf A}_{m}(Q){\bf x}_{m}],\\
  \frac{\partial{\bf A}_{m}{\bf x}_{m}}{\partial{\bm\alpha}_{I}}
  &= j[\breve{\bf A}_{m}(1){\bf x}_{m},\ldots,\breve{\bf A}_{m}(Q){\bf x}_{m}],
\end{align}
\end{subequations}
where $\breve{\bf A}_{m}(q)={\bf A}_{\mathrm{R},m}{\bf \Lambda}_{q}{\bf A}_{\mathrm{T},m}^{T}$ and ${\bf \Lambda}_{q}$ is the selection matrix whose $q$-th diagonal entry is one and the others are zero.
Substituting \eqref{FIM_RCS_NB} into the narrow-band FIM expression yields
\begin{small}
\begin{align}\label{CRB_RCS_NB}
&\hspace*{-0.5em}\mathrm{CRB}_{\alpha_{R_q}}\!=\!\mathrm{CRB}_{\alpha_{I_q}}
\!=\!\frac{\sigma_r^2}{2\,\mathbb{E}\!\left\{\sum_{m=1}^{M}\left\|
{\bf a}_{\mathrm{R}}(m,q){\bf a}_{\mathrm{T}}^{T}(m,q){\bf x}_{m}
\right\|_2^2\right\}}.
\end{align}
\end{small}
%
\vspace{-0.2cm}
\subsection{Conditional CRB for Velocity Estimation}
\vspace{-0.1cm}
Similarly, the velocity-related derivatives are given by
\begin{subequations}\label{Derivation_vel_NB}
\begin{align}
  \frac{\partial{\bf A}_{m}{\bf x}_{m}}{\partial{\bf v}_x}
  &= [\acute{\bf A}_{m}(1){\bf x}_{m},\ldots,\acute{\bf A}_{m}(Q){\bf x}_{m}], \\
  \frac{\partial{\bf A}_{m}{\bf x}_{m}}{\partial{\bf v}_y}
  &= [\grave{\bf A}_{m}(1){\bf x}_{m},\ldots,\grave{\bf A}_{m}(Q){\bf x}_{m}],
\end{align}
\end{subequations}
where
\begin{align*}
\acute{\bf A}_{m}(q)
&=\big[{\acute{\bf A}_{\mathrm{R},m}}{\bf \Lambda}_{q}{\bf B}{\bf A}_{\mathrm{T},m}^{T}
+{\bf A}_{\mathrm{R},m}{\bf B}{\bf \Lambda}_{q}{\acute{\bf A}_{\mathrm{T},m}^{T}}\big],\\
\grave{\bf A}_{m}(q)
&=\big[{\grave{\bf A}_{\mathrm{R},m}}{\bf \Lambda}_{q}{\bf B}{\bf A}_{\mathrm{T},m}^{T}
+{\bf A}_{\mathrm{R},m}{\bf B}{\bf \Lambda}_{q}{\grave{\bf A}_{\mathrm{T},m}^{T}}\big],
\end{align*}
denote the partial derivative of channel matrix w.r.t $v_{x_q}$ and $v_{y_q}$, respectively, and  each element of the steering-vector derivative is (for $\ast\in\{n_t,n_r\}$)
\begin{align*}
\frac{\partial{\bf a}_{\mathrm{R},\mathrm{T},\ast}(m,q)}{\partial v_{x_q}}
&=
\frac{j2\pi\frac{f_c}{c}\,mT_{\mathrm{sym}}(x_q-x_{\ast})}{\|{\mathbb I}_q-{\mathbb I}_{\ast}\|}
\,{\bf a}_{\mathrm{R},\mathrm{T},\ast}(m,q),\\
\frac{\partial{\bf a}_{\mathrm{R},\mathrm{T},\ast}(m,q)}{\partial v_{y_q}}
&=
\frac{j2\pi\frac{f_c}{c}\,mT_{\mathrm{sym}}(y_q-y_{\ast})}{\|{\mathbb I}_q-{\mathbb I}_{\ast}\|}
\,{\bf a}_{\mathrm{R},\mathrm{T},\ast}(m,q).
\end{align*}
Substituting \eqref{Derivation_vel_NB} into the narrow-band FIM \eqref{FIM_block_NB}, the CRB for velocity estimation as
\begin{subequations}\label{CRB_vel_q_NB}
\begin{align}
 \mathrm{CRB}_{v_{x_q}}
 &=\frac{\sigma_r^2}{2\,\mathbb{E}\!\left\{\sum_{m=1}^{M}\|\acute{\bf A}_{m}(q){\bf x}_{m}\|_2^2\right\}},\\
 \mathrm{CRB}_{v_{y_q}}
 &=\frac{\sigma_r^2}{2\,\mathbb{E}\!\left\{\sum_{m=1}^{M}\|\grave{\bf A}_{m}(q){\bf x}_{m}\|_2^2\right\}}.
\end{align}
\end{subequations}
%
\subsection{Conditional CRB for Location Estimation}
\vspace{-0.1cm}
Meanwhile, the location-related derivatives are
\begin{subequations}\label{Derivation_loc_NB}
\begin{align}
  \frac{\partial{\bf A}_{m}{\bf x}_{m}}{\partial{\bf x}}
  &= [\dot{\bf A}_{m}(1){\bf x}_{m},\ldots,\dot{\bf A}_{m}(Q){\bf x}_{m}],\\
  \frac{\partial{\bf A}_{m}{\bf x}_{m}}{\partial{\bf y}}
  &= [\ddot{\bf A}_{m}(1){\bf x}_{m},\ldots,\ddot{\bf A}_{m}(Q){\bf x}_{m}],
\end{align}
\end{subequations}
where
\begin{align*}
\dot{\bf A}_{m}(q)
&=\big[{\dot{\bf A}_{\mathrm{R},m}}{\bf \Lambda}_{q}{\bf B}{\bf A}_{\mathrm{T},m}^{T}
+{\bf A}_{\mathrm{R},m}{\bf B}{\bf \Lambda}_{q}{\dot{\bf A}_{\mathrm{T},m}^{T}}\big],\\
\ddot{\bf A}_{m}(q)
&=\big[{\ddot{\bf A}_{\mathrm{R},m}}{\bf \Lambda}_{q}{\bf B}{\bf A}_{\mathrm{T},m}^{T}
+{\bf A}_{\mathrm{R},m}{\bf B}{\bf \Lambda}_{q}{\ddot{\bf A}_{\mathrm{T},m}^{T}}\big].
\end{align*}
Each element of the derivative w.r.t.\ $x_q$ and $y_q$ can be written as
$\frac{\partial{\bf a}_{\mathrm{R},\mathrm{T},\ast}(m,q)}{\partial x_q}
= \dot g_{\mathrm{R},\mathrm{T},\ast}(m,q)\,{\bf a}_{\mathrm{R},\mathrm{T},\ast}(m,q)$
and
$\frac{\partial{\bf a}_{\mathrm{R},\mathrm{T},\ast}(m,q)}{\partial y_q}
= \ddot g_{\mathrm{R},\mathrm{T},\ast}(m,q)\,{\bf a}_{\mathrm{R},\mathrm{T},\ast}(m,q)$,
where $\dot g_{\mathrm{R},\mathrm{T},\ast}(m,q)$ and $\ddot g_{\mathrm{R},\mathrm{T},\ast}(m,q)$
by \eqref{derivation_x} arranged at the top of next page
\begin{figure*}[!t]
\begin{small}
\begin{subequations} \label{derivation_x}
\begin{align} 
  &\dot{g}_{\mathrm{R},\mathrm{T},\ast}(m,q) \!=\!\left(\frac{j2\pi\frac{f_c}{c}(v_{x}^{(q)}mT_{\mathrm{sym}}\!-\!(x_q\!-\!{x}_{\ast}))}{\|{\mathbb I}_{q}-{\mathbb I}_{\ast}\|}-\frac{{x}_q-{x}_{\ast}}{{\|{\mathbb I}_{q}-{\mathbb I}_{\ast}\|^2}} -\frac{j2\pi\frac{f_c}{c}mT_{\mathrm{sym}}(v_{x}^{(q)}({x}_q-{x}_{\ast})^2\!+\!v_{y}^{(q)}({x}_q\!-\!{x}_{\ast})({y}_q\!-\!{y}_{\ast}))}{\|{\mathbb I}_{q}-{\mathbb I}_{\ast}\|^3}\right). & \\
  &\ddot{g}_{\mathrm{R},\mathrm{T},\ast}(m,q) \!=\!\left(\frac{j2\pi\frac{f_c}{c}(v_{y}^{(q)}mT_{\mathrm{sym}}\!-\!(y_q\!-\!{y}_{\ast}))}{\|{\mathbb I}_{q}-{\mathbb I}_{\ast}\|}-\frac{{y}_q-{y}_{\ast}}{{\|{\mathbb I}_{q}-{\mathbb I}_{\ast}\|^2}} -\frac{j2\pi\frac{f_c}{c}mT_{\mathrm{sym}}(v_{y}^{(q)}({y}_q-{y}_{\ast})^2\!+\!v_{x}^{(q)}({x}_q\!-\!{x}_{\ast})({y}_q\!-\!{y}_{\ast}))}{\|{\mathbb I}_{q}-{\mathbb I}_{\ast}\|^3}\right). &  
\end{align}
\end{subequations}
\end{small}
\hrulefill
\end{figure*}

Substituting \eqref{Derivation_loc_NB} into the narrow-band \eqref{FIM_block_NB}, we have the CRB for location estimation as
\begin{subequations}\label{CRB_loc_q_NB}
\begin{align}
 \mathrm{CRB}_{x_q}
 &=\frac{\sigma_r^2}{2\,\mathbb{E}\!\left\{\sum_{m=1}^{M}\|\dot{\bf A}_{m}(q){\bf x}_{m}\|_2^2\right\}},\\
 \mathrm{CRB}_{y_q}
 &=\frac{\sigma_r^2}{2\,\mathbb{E}\!\left\{\sum_{m=1}^{M}\|\ddot{\bf A}_{m}(q){\bf x}_{m}\|_2^2\right\}}.
\end{align}
\end{subequations}
Recalling \eqref{iso_trans}, we can simplify and rewrite the CRB for RCS, velocity, and location as \eqref{CRB_q_single} which arranged at the top of next page. 
\begin{figure*}[!t]
\begin{small}
\begin{subequations} \label{CRB_q_single}
\begin{flalign}
&\mathrm{CRB}_{\alpha_{R_{q}}} \!=\!\mathrm{CRB}_{\alpha_{I_{q}}} =\frac{\sigma_r^2}{2\mathcal{P}{\sum}_{k=1}^K{\sum}_{m=1}^M\|{\bf a}_{\mathrm{R}}(m,q)\|_2^2\|{\bf a}_{\mathrm{T}}^{T}(m,q)\|_2^2},\\
 &\mathrm{CRB}_{v_{x}} \!=\!\frac{\sigma_r^2}{2|\alpha|^2\{\sum_{m=1}^M    (\|\acute{\bf a}_{\mathrm R}(m)\|_2^2\|{\bf a}_{\mathrm T}(m)\|_2^2\!+\!2\Re{\{\acute{\bf a}_{\mathrm R}^H(m){\bf a}_{\mathrm R}(m)}{\bf a}_{\mathrm T}^H(m)\acute{\bf a}_{\mathrm T}(m)\}\!+\!\|{\bf a}_{\mathrm R}(m)\|_2^2\|\acute{\bf a}_{\mathrm T}(m)\|_2^2) \}},     \label{CRB_vel.q1_single}\\
 &\mathrm{CRB}_{v_{y}}  \!=\!\frac{\sigma_r^2}{2|\alpha|^2\{\sum_{m=1}^M    (\|\grave{\bf a}_{\mathrm R}(m)\|_2^2\|{\bf a}_{\mathrm T}(m)\|_2^2\!+\!2\Re{\{\grave{\bf a}_{\mathrm R}^H(m){\bf a}_{\mathrm R}(m)}{\bf a}_{\mathrm T}^H(m)\grave{\bf a}_{\mathrm T}(m)\}\!+\!\|{\bf a}_{\mathrm R}(m)\|_2^2\|\grave{\bf a}_{\mathrm T}(m)\|_2^2) \}}, \label{CRB_vel.q2_single} \\
 &\mathrm{CRB}_{x} \!=\! \frac{\sigma_r^2}{2|\alpha|^2\{\sum_{m=1}^M    (\|\dot{\bf a}_{\mathrm R}(m)\|_2^2\|{\bf a}_{\mathrm T}(m)\|_2^2+2\Re{\{\dot{\bf a}_{\mathrm R}^H(m){\bf a}_{\mathrm R}(m)}{\bf a}_{\mathrm T}^H(m)\dot{\bf a}_{\mathrm T}(m)\}+\|{\bf a}_{\mathrm R}(m)\|_2^2\|\dot{\bf a}_{\mathrm T}(m)\|_2^2) \}},     \label{CRB_loc.q1_single}\\
 &\mathrm{CRB}_{y} \!=\!\frac{\sigma_r^2}{2|\alpha|^2\{\sum_{m=1}^M    (\|\ddot{\bf a}_{\mathrm R}(m)\|_2^2\|{\bf a}_{\mathrm T}(m)\|_2^2+2\Re{\{\ddot{\bf a}_{\mathrm R}^H(m){\bf a}_{\mathrm R}(m)}{\bf a}_{\mathrm T}^H(m)\ddot{\bf a}_{\mathrm T}(m)\}+\|{\bf a}_{\mathrm R}(m)\|_2^2\|\ddot{\bf a}_{\mathrm T}(m)\|_2^2) \}}. \label{CRB_loc.q2_single}
\end{flalign}
\end{subequations}
\end{small}
\hrulefill
\end{figure*}
With the general CRB expressions given in \eqref{CRB_q_single}, we can then explore the narrow-band near-field CRB for the special cases. 
\vspace{-0.2cm}
\subsection{Approximated CRB for Single Target Case}
\vspace{-0.1cm}
Most existing CRB results in the far-field MIMO radar literature and in narrowband near-field localization primarily focus on simplified settings, e.g., static objects or a subset of parameters \cite{liang2021CramerRao}, which limits their applicability to dynamic near-field sensing scenarios. Motivated by this, this section investigates tractable CRB approximations for the moving target under a narrowband near-field snapshot model. Hereafter, we use the abbreviations FF and NF to denote the far-field and near-field regimes, respectively.
\subsubsection{RCS Estimation}
We can first define the gains of Tx and Rx, respectively, as
\begin{subequations} \label{gain_Tx_Rx}
 \begin{align} 
   G_{\mathrm{Tx}} &\!=\! \|{\bf a}_{\mathrm{T}}(m,q)\|_2^2 = \frac{\lambda_c^2}{16\pi^2}\sum_{n_t=1}^{N_t}\frac{1}{r_{n_t}^2} = \frac{\lambda_c^2}{16\pi^2}g_{\mathrm {Tx}}, \label{gain_Tx}\\ 
   G_{\mathrm{Rx}} &\!=\! \|{\bf a}_{\mathrm{R}}(m,q)\|_2^2 = \frac{\lambda_c^2}{16\pi^2}\sum_{n_r=1}^{N_r}\frac{1}{r_{n_r}^2} = \frac{\lambda_c^2}{16\pi^2}g_{\mathrm {Rx}}, \label{gain_Rx}
 \end{align}   
\end{subequations}
where $r_{n_t}=\|{\mathbb I}_{q}-{\mathbb I}_{n_t}\|$ and $r_{n_r}=\|{\mathbb I}_{q}-{\mathbb I}_{n_r}\|$. 
Then, we have 
\begin{flalign} \label{CRB_RCS}
\mathrm{CRB}_{\alpha_{R}}=\mathrm{CRB}_{\alpha_{I}}=\frac{256\sigma_r^2\pi^{4}}{2\mathcal{P}Mg_{\mathrm {Tx}}g_{\mathrm {Rx}}\lambda_c^4}. 
\end{flalign} 
Moreover, we denote $\mathrm{CRB}_{\alpha}=\mathrm{CRB}_{\alpha_{R}}+\mathrm{CRB}_{\alpha_{I}}$ by the CRB for RCS estimation. 
\begin{figure}[t]
\centering{\includegraphics[width=1\columnwidth]{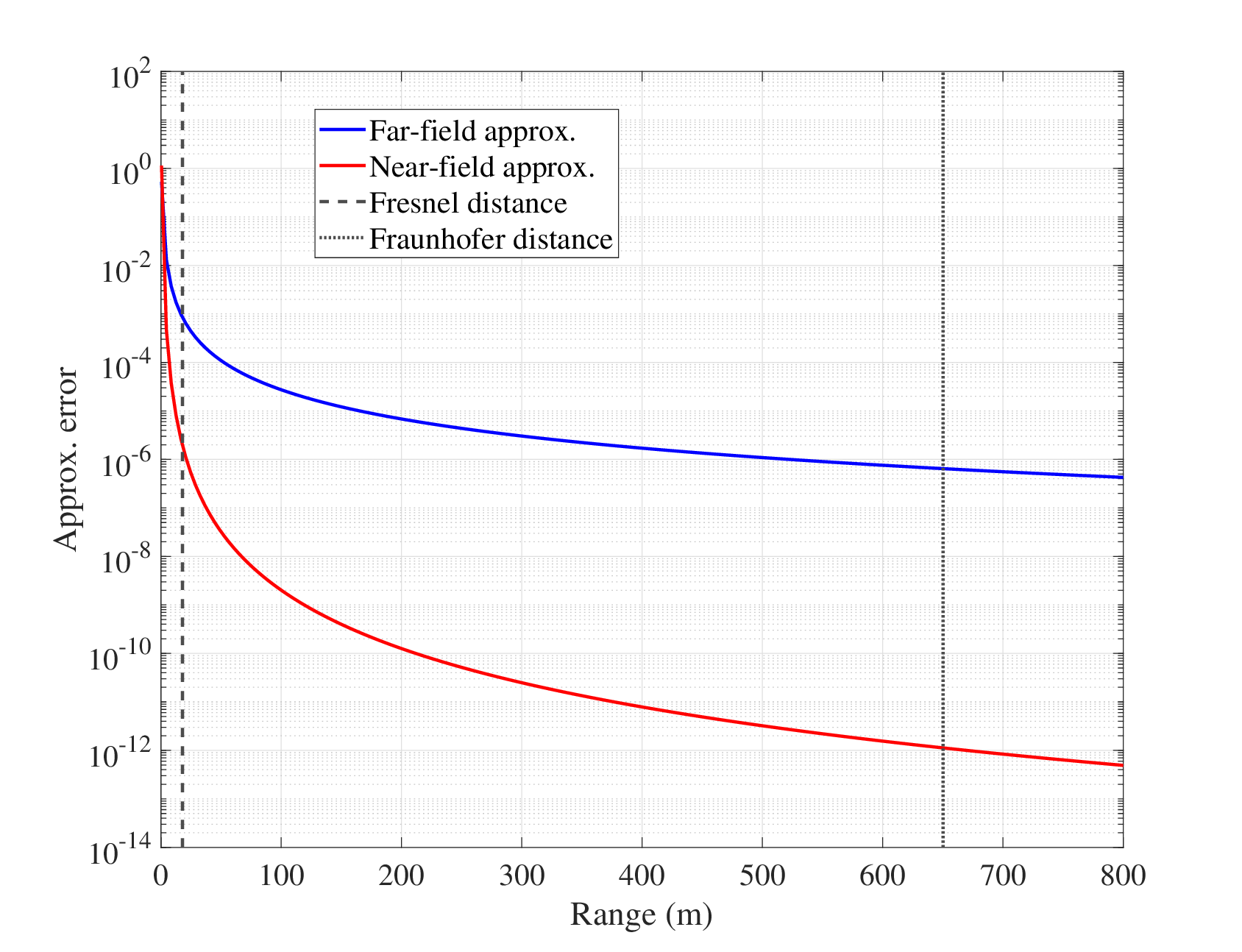}}
\vspace{-0.2cm}\caption{Approximation error versus target range, ($N_t=256$).
\label{range}}
\end{figure}
\begin{figure}[t]
\centering{\includegraphics[width=1\columnwidth]{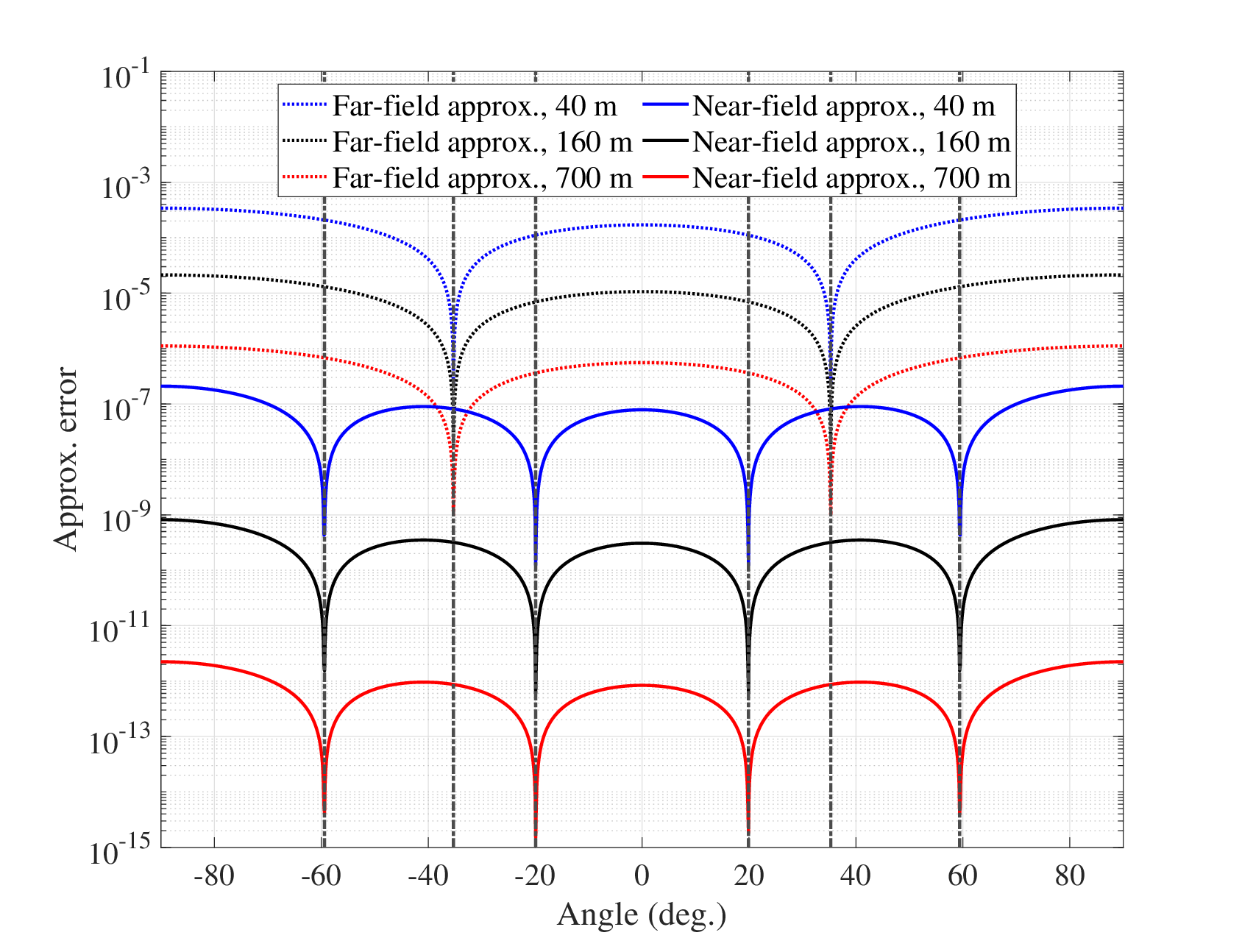}}
\vspace{-0.2cm}\caption{Approximation error versus target angle, ($N_t=256$)
\label{angle}}
\end{figure}
\begin{theorem}
 For a single target with co-located Tx/Rx arrays along with $x$-axis centered at the origin, CRBs for RCS estimation, in terms of far-field and near-field cases, are, respectively, given by
\begin{subequations} \label{CRB_RCS.1}
\begin{flalign} 
&\mathrm{CRB}_{\alpha}^{\mathrm{FF}}  \approx \frac{256\sigma_r^2\pi^{4}(r_{\mathrm{Tx}}^{\mathrm{FF}}r_{\mathrm{Rx}}^{\mathrm{FF}})^2}{\mathcal{P}M{N_t}{N_r}\lambda_c^4}, \label{CRB_RCS.1.1}\\
&\mathrm{CRB}_{\alpha}^{\mathrm{NF}} \! \approx \!\frac{256\sigma_r^2\pi^{4}(r_{\mathrm{Tx}}^{\mathrm{NF}}r_{\mathrm{Rx}}^{\mathrm{NF}})^2}{\mathcal{P}M{N_t}{N_r}(1+\Delta_{\mathrm{Tx}})(1+\Delta_{\mathrm{Rx}})\lambda_c^4}, \label{CRB_RCS.1.2}
\end{flalign} 
\end{subequations}
where  $r_{\mathrm{Tx}}^{\mathrm{FF}}$ and $r_{\mathrm{Rx}}^{\mathrm{FF}}$ denote the range between the far-field target and reference element of Tx and Rx, respectively, and $\Delta_{\mathrm{Tx}}=\frac{(N_t^{2}-1)d_t^2(4\sin^2\theta_{\mathrm{Tx},q} - 1)}{12 r_{\mathrm{Tx}}^2}$ and $\Delta_{\mathrm{Rx}}=\frac{(N_r^{2}-1)d_r^2(4\sin^2\theta_{\mathrm{Rx},q} - 1)}{12 r_{\mathrm{Rx}}^2}$. 
\end{theorem}
\begin{proof}
  See Appendix A. 
\end{proof}

Fig.~\ref{range} and Fig~\ref{angle} display the far-field (first-order Taylor) and near-field (second-order Taylor) approximation error to the true $g_{\mathrm {Tx}}$ or $g_{\mathrm {Rx}}$ in \eqref{gain_Tx_Rx}. Fig.~\ref{range} shows that the error decreases with target range. Meanwhile, the NF approximation is consistently much more accurate, and far-field becomes reliable mainly beyond the Fraunhofer distance. Fig.~\ref{angle} shows that the error depends strongly on the angle and shows symmetric dips. Then, the larger range reduces error and near-field stays better, especially at short or medium ranges.
\subsubsection{Velocity Estimation}
As for the CRB of velocity estimation, we have the following Theorem.
\begin{theorem}
Then, for the far-field target, we omit the second-order items $\frac{(N_t^2-1)d_t^2}{12{r_{\mathrm{Tx}}^{\mathrm{FF}}}^2}$ and $\frac{(N_r^2-1)d_r^2}{12{r_{\mathrm{Rx}}^{\mathrm{FF}}}^2}$, and have
\begin{small}
\begin{subequations}
\label{CRB_vel_FF_final}
\begin{align}
\mathrm{CRB}_{v_x}^{\mathrm{FF}}
&\!\approx\!
\frac{32\pi^2\sigma_r^2(r_{\mathrm{Tx}}^{\mathrm{FF}}r_{\mathrm{Rx}}^{\mathrm{FF}})^2}
{|\alpha|^2\mathcal{P}N_tN_rT_{\mathrm{sym}}^2
C_M\big(\sin\theta_{\mathrm{Tx},q}\!+\!\sin\theta_{\mathrm{Rx},q}\big)^2\lambda_c^2}, \\
\mathrm{CRB}_{v_y}^{\mathrm{FF}} 
&\!\approx\!
\frac{32\pi^2\sigma_r^2(r_{\mathrm{Tx}}^{\mathrm{FF}}r_{\mathrm{Rx}}^{\mathrm{FF}})^2}
{|\alpha|^2\mathcal{P}N_tN_rT_{\mathrm{sym}}^2
C_M
\big(\cos\theta_{\mathrm{Tx},q}\!+\!\cos\theta_{\mathrm{Rx},q}\big)^2\lambda_c^2}.
\end{align}
\end{subequations}
\end{small}
\noindent\hspace{-0.8em} where $C_M=\frac{M(M+1)(2M+1)}{6}$.
Then, for the near-field target, we have 
\begin{small}
\begin{subequations}\label{CRB_vel_NF_final}
\begin{align}
\hspace*{-1em}\mathrm{CRB}_{v_x}^{\mathrm{NF}} 
\!\approx\!&
\frac{32\pi^2\sigma_r^2({r_{\mathrm{Tx}}^{\mathrm{NF}}}{r_{\mathrm{Rx}}^{\mathrm{NF}}})^2}
{|\alpha|^2\mathcal{P}N_tN_rT_{\mathrm{sym}}^2
C_M\big(\sin\theta_{\mathrm{Tx},q}\!+\!\sin\theta_{\mathrm{Rx},q}\big)^2
\Psi_{x}\lambda_c^2}, \\
\hspace*{-1em}\mathrm{CRB}_{v_y}^{\mathrm{NF}}
\!\approx\!&
\frac{32\pi^2\sigma_r^2({r_{\mathrm{Tx}}^{\mathrm{NF}}}{r_{\mathrm{Rx}}^{\mathrm{NF}}})^2}
{|\alpha|^2\mathcal{P}N_tN_rT_{\mathrm{sym}}^2
C_M
\big(\cos\theta_{\mathrm{Tx},q}\!+\!\cos\theta_{\mathrm{Rx},q}\big)^2
\Psi_{y}\lambda_c^2}, 
\end{align}
\end{subequations}
\end{small}
\noindent\hspace{-1em} where the near-field correction terms $\Psi_{x}$ and $\Psi_{y}$ are given by
\begin{subequations} \label{correct_term_v}
\begin{align}
\Psi_x  \!=\!\frac{\Phi_x}{\left(\sin\theta_{\mathrm{Tx},q}+\sin\theta_{\mathrm{Rx},q}\right)^2}, \label{correct_term_vx}\\
\Psi_y  \!=\!\frac{\Phi_y}{\left(\cos\theta_{\mathrm{Tx},q}+\cos\theta_{\mathrm{Rx},q}\right)^2}, \label{correct_term_vy}
\end{align}
\end{subequations}
where the numerators $\Phi_x$  and $\Phi_y$ are given by 
\begin{small}
\begin{subequations} \label{correct_sub_term_v}
\begin{align}
\Phi_x &\!\triangleq\! A_{\mathrm{Tx}} B_{\mathrm{Rx}}^{(x)} \!+\! A_{\mathrm{Rx}} B_{\mathrm{Tx}}^{(x)}
\!+\! 2\sin\theta_{\mathrm{Tx},q}\sin\theta_{\mathrm{Rx},q}
\left(1\!+\!\Delta_{x,\mathrm{Tx}}^{\mathrm{NF}}\!+\!\Delta_{x,\mathrm{Rx}}^{\mathrm{NF}}\right), \label{correct_sub_term_vx}\\
\Phi_y &\!\triangleq\! A_{\mathrm{Tx}} B_{\mathrm{Rx}}^{(y)} \!+\! A_{\mathrm{Rx}} B_{\mathrm{Tx}}^{(y)}
\!+\! 2\cos\theta_{\mathrm{Tx},q}\cos\theta_{\mathrm{Rx},q}
\left(1\!+\!\Delta_{y,\mathrm{Tx}}^{\mathrm{NF}}\!+\!\Delta_{y,\mathrm{Rx}}^{\mathrm{NF}}\right), \label{correct_sub_term_vy}
\end{align}
\end{subequations}
\end{small}
in which the variables $A_{\mathrm{Tx}}$ and  $A_{\mathrm{Rx}}$ are given by
\begin{subequations}
\begin{align}
A_{\mathrm{Tx}} = 1 + \frac{\left(N_t^{2}-1\right)d_t^{2}}{12r_{\mathrm{Tx}}^{2}}
\left(4\sin^{2}\theta_{\mathrm{Tx},q}-1\right),\\
A_{\mathrm{Rx}} = 1 + \frac{\left(N_r^{2}-1\right)d_r^{2}}{12r_{\mathrm{Rx}}^{2}}
\left(4\sin^{2}\theta_{\mathrm{Rx},q}-1\right),
\end{align}
\end{subequations}
the variables related to $B_{\mathrm{Tx}}$ and $B_{\mathrm{Rx}}$ are given by
\begin{small}
\begin{subequations}
\begin{align}
B_{\mathrm{Tx}}^{(x)} 
&\!=\! \sin^{2}\theta_{\mathrm{Tx},q}
\!+\! \frac{\left(N_t^{2}\!-\!1\right)d_t^{2}}{12r_{\mathrm{Tx}}^{2}}
\left(12\sin^{4}\theta_{\mathrm{Tx},q}\!-\!10\sin^{2}\theta_{\mathrm{Tx},q}\!+\!1\right), \\
B_{\mathrm{Rx}}^{(x)} 
&\!=\! \sin^{2}\theta_{\mathrm{Rx},q}
\!+\! \frac{\left(N_r^{2}\!-\!1\right)d_r^{2}}{12r_{\mathrm{Rx}}^{2}}
\left(12\sin^{4}\theta_{\mathrm{Rx},q}\!-\!10\sin^{2}\theta_{\mathrm{Rx},q}\!+\!1\right), \\
B_{\mathrm{Tx}}^{(y)} 
&\!=\! \cos^{2}\theta_{\mathrm{Tx},q}
\!+\! \frac{\left(N_t^{2}\!-\!1\right)d_t^{2}}{12r_{\mathrm{Tx}}^{2}}
\cos^{2}\theta_{\mathrm{Tx},q}\left(12\sin^{2}\theta_{\mathrm{Tx},q}\!-\!2\right), \\
B_{\mathrm{Rx}}^{(y)} 
&\!=\! \cos^{2}\theta_{\mathrm{Rx},q}
\!+\! \frac{\left(N_r^{2}\!-\!1\right)d_r^{2}}{12r_{\mathrm{Rx}}^{2}}
\cos^{2}\theta_{\mathrm{Rx},q}\left(12\sin^{2}\theta_{\mathrm{Rx},q}\!-\!2\right),
\end{align}
\end{subequations}
\end{small}
and the variables related to the $\Delta^{\mathrm{NF}}$ are listed as 
\begin{subequations}
\begin{align}
 \Delta_{x,\mathrm{Tx}}^{\mathrm{NF}} \triangleq \frac{\left(N_t^{2}-1\right)d_t^{2}}{8r_{\mathrm{Tx}}^{2}}
\left(-3+5\sin^{2}\theta_{\mathrm{Tx},q}\right),\\
\Delta_{x,\mathrm{Rx}}^{\mathrm{NF}} \triangleq \frac{\left(N_r^{2}-1\right)d_r^{2}}{8r_{\mathrm{Rx}}^{2}}
\left(-3+5\sin^{2}\theta_{\mathrm{Rx},q}\right),  \\
\Delta_{y,\mathrm{Tx}}^{\mathrm{NF}} \triangleq \frac{\left(N_t^{2}-1\right)d_t^{2}}{8r_{\mathrm{Tx}}^{2}}
\left(-1+5\sin^{2}\theta_{\mathrm{Tx},q}\right),\\
\Delta_{y,\mathrm{Rx}}^{\mathrm{NF}} \triangleq \frac{\left(N_r^{2}-1\right)d_r^{2}}{8r_{\mathrm{Rx}}^{2}}
\left(-1+5\sin^{2}\theta_{\mathrm{Rx},q}\right).
\end{align}
\end{subequations}
\end{theorem}

\begin{proof}
    See appendix B.
\end{proof}

Note that the correction terms $\Psi_{x}$ and $\Psi_{y}$ will be closed to 1 when the range of the target is larger than Fraunhofer distance.
\subsubsection{Location Estimation} Finnaly, for the CRB of location estimation, we have the following Theorem.
\begin{theorem}
For the far-field target, we have
\begin{small}
\begin{subequations} \label{FF_approx_loc}
\begin{align}
\dot{G}^{\mathrm{FF}}_{\mathrm{Tx}} 
&\approx \frac{N_t}{4 (r^{\mathrm{FF}}_{\mathrm{Tx}})^{2}} \sin^{2}\theta_{\mathrm{Tx},q},
\dot{G}^{\mathrm{FF}}_{\mathrm{Rx}} 
&\approx \frac{N_r}{4 (r^{\mathrm{FF}}_{\mathrm{Rx}})^{2}} \sin^{2}\theta_{\mathrm{Rx},q}, \\
\ddot{G}^{\mathrm{FF}}_{\mathrm{Tx}}
&\approx \frac{N_t}{4 (r^{\mathrm{FF}}_{\mathrm{Tx}})^{2}} \cos^{2}\theta_{\mathrm{Tx},q},
\ddot{G}^{\mathrm{FF}}_{\mathrm{Rx}}
&\approx \frac{N_r}{4 (r^{\mathrm{FF}}_{\mathrm{Rx}})^{2}} \cos^{2}\theta_{\mathrm{Rx},q}.
\end{align}
\end{subequations}
\end{small}
\noindent\hspace{-0.5em}Substituting \eqref{FF_approx_loc} into \eqref{CRB_loc.q1_single} and  \eqref{CRB_loc.q2_single}, we have the far-field wide-band CRB for localization as  
\begin{small}
\begin{subequations}
\label{CRB_loc_FF_NB}
\begin{align}
\mathrm{CRB}_{x}^{\mathrm{FF}}
&\approx
\frac{32\pi^2\sigma_r^2
\big(r_{\mathrm{Tx}}^{\mathrm{FF}}r_{\mathrm{Rx}}^{\mathrm{FF}}\big)^2}
{|\alpha|^2\mathcal PMN_tN_r
\Big(\sin\theta_{\mathrm{Tx},q}+\sin\theta_{\mathrm{Rx},q}\Big)^2
\lambda_c^2},\\
\mathrm{CRB}_{y}^{\mathrm{FF}}
&\approx
\frac{32\pi^2\sigma_r^2
\big(r_{\mathrm{Tx}}^{\mathrm{FF}}r_{\mathrm{Rx}}^{\mathrm{FF}}\big)^2}
{|\alpha|^2\mathcal PMN_tN_r
\Big(\cos\theta_{\mathrm{Tx},q}+\cos\theta_{\mathrm{Rx},q}\Big)^2
\lambda_c^2}.
\end{align}
\end{subequations}
\end{small}
Similarly, for the near-field target, we have 
\begin{small}
\begin{subequations}
\label{CRB_loc_NF_NB}
\begin{align}
\mathrm{CRB}_{x}^{\mathrm{NF}}
&\approx
\frac{32\pi^2\sigma_r^2
\big(r_{\mathrm{Tx}}^{\mathrm{NF}}r_{\mathrm{Rx}}^{\mathrm{NF}}\big)^2}
{|\alpha|^2\mathcal PMN_tN_r
\Big(\sin\theta_{\mathrm{Tx},q}+\sin\theta_{\mathrm{Rx},q}\Big)^2
\Psi_x\lambda_c^2},\\
\mathrm{CRB}_{y}^{\mathrm{NF}}
&\approx
\frac{32\pi^2\sigma_r^2
\big(r_{\mathrm{Tx}}^{\mathrm{NF}}r_{\mathrm{Rx}}^{\mathrm{NF}}\big)^2}
{|\alpha|^2\mathcal PMN_tN_r
\Big(\cos\theta_{\mathrm{Tx},q}+\cos\theta_{\mathrm{Rx},q}\Big)^2
\Psi_y\lambda_c^2}.
\end{align}
\end{subequations}
\end{small} 
\end{theorem}
\begin{proof}
The proof is analogous to that of the velocity CRB in Appendix B and hence is omitted for brevity.
The only difference is that the steering-vector derivatives with respect to the location parameters $(x_q,y_q)$ do not introduce the slow-time factor $mT_{\mathrm{sym}}$ appearing in the velocity derivatives.
Consequently, the slow-time accumulation reduces to
$\sum_{m=1}^{M}1=M$ (instead of $\sum_{m=1}^{M}m^{2}T_{\mathrm{sym}}^{2}=T_{\mathrm{sym}}^{2}C_M$),
while the remaining frequency-domain summation and the subsequent far-field and near-field approximations proceed identically.
\end{proof}

\begin{remark}\textbf{Multi-target Case:}
For $Q>1$ targets, the FIM becomes block-structured, where the diagonal blocks correspond to the single-target information and the off-diagonal blocks capture inter-target coupling due to non-orthogonal steering responses. The CRB of target $q$ is given by the $q$-th diagonal block of the inverse FIM, i.e., $\mathrm{CRB}(\boldsymbol{\theta}_q)=\big[\mathbf{J}^{-1}(\boldsymbol{\theta})\big]_{qq}$. When targets are well separated (or approximately orthogonal), the coupling blocks are negligible and the multi-target CRBs is the summation of each single-target bound. 
\end{remark}

\section{Numerical Experiments}
Throughout the simulations, the system parameters are configured according to the potential 6G frequency bands identified by the ITU-R and recent studies on upper-mid (FR3) spectrum allocation \cite{testolina2024sharing,emil2025enabling}.
Specifically, we set carrier frequency $f_c=15$~GHz (wavelength $\lambda_c=c/f_c$ and $c=3\times10^8$~m/s) . The number of snapshot is set as 
$M=256$. Both Tx and Rx employ centered ULAs along the $x$-axis with $N_t=N_r=256$ elements and half-wavelength spacing $d_t=d_r=\lambda_c/2$, yielding apertures $D_{\mathrm{Tx}}=(N_t-1)d_t$ and $D_{\mathrm{Rx}}=(N_r-1)d_r$. The reactive and radiative near-field boundaries are computed via the standard expressions $0.62\sqrt{D^3/\lambda_c}$ and $2D^2/\lambda_c$ for both arrays \cite{elbir2023near-field,elbir2024spherical,elbir2023nba,elbir2024near,hu2024joint}.
Additive noise is modeled as spatially and temporally white circular Gaussian with variance ${\sigma}_r^2 = -114$ dBm. Transmit power is fixed to $\mathcal{P}=0.1$ W. For the single target case, we set the RCS and the velocity as ${\alpha}=1+0.1j$ and  ${\bf v}=[1,4]$ m/s, respectively. The relative error is calculated by $\epsilon_{\mathrm{rel}} =\left|\left(\mathrm{CRB}_{\mathrm{approx.}}-\mathrm{CRB}_{\mathrm{true}}\right)/\mathrm{CRB}_{\mathrm{true}}\right|$. Throughout the simulations, the true CRB is directly obtained from \eqref{CRB_RCS_NB}, \eqref{CRB_vel_q_NB} and \eqref{CRB_loc_q_NB}, while its far-field and near-field approximations are obtained by the corresponding expressions in \eqref{CRB_RCS.1}, \eqref{CRB_vel_FF_final}, \eqref{CRB_vel_NF_final}, \eqref{CRB_loc_FF_NB}, and \eqref{CRB_loc_NF_NB}, respectively.
\vspace{-0.2cm}
\subsection{CRB for RCS with Single Target and Monostatic Setting}
We first consider a single-target monostatic setup, where the Tx and Rx arrays are co-located. In this case, the bistatic range reduces to twice the target distance and the measurement model becomes symmetric with respect to the array center i.e., $[0,0]$. 
Fig.~\ref{RCS_range} illustrates the CRB for RCS estimation versus the target range (with $N_t=256$, $\theta=20^{\circ}$). As expected, the CRB increases monotonically with range. This trend is primarily driven by the two-way propagation attenuation: in the monostatic near-field model, the effective channel amplitude contains the product of Tx/Rx distance-dependent gains, so the Fisher information decreases rapidly as the target moves farther away, leading to a steadily looser bound on the RCS estimate.

Fig.~\ref{RCS_range} also compares the proposed near-field (NF) approximation and the far-field (FF) approximation against the true CRB. The NF approximation is essentially closer to the true CRB across the entire range sweep, and its relative error remains consistently negligible (right axis). In contrast, the FF approximation exhibits a noticeable mismatch at short ranges where wavefront curvature is significant; however, its accuracy improves as the range increases and the propagation progressively approaches the plane-wave regime. These results confirm that the NF approximation is necessary to correctly predict RCS estimation limits in the radiative near-field, while far-field simplifications become reliable only at sufficiently large ranges. 
\begin{figure}[t]
\centering{\includegraphics[width=1\columnwidth]{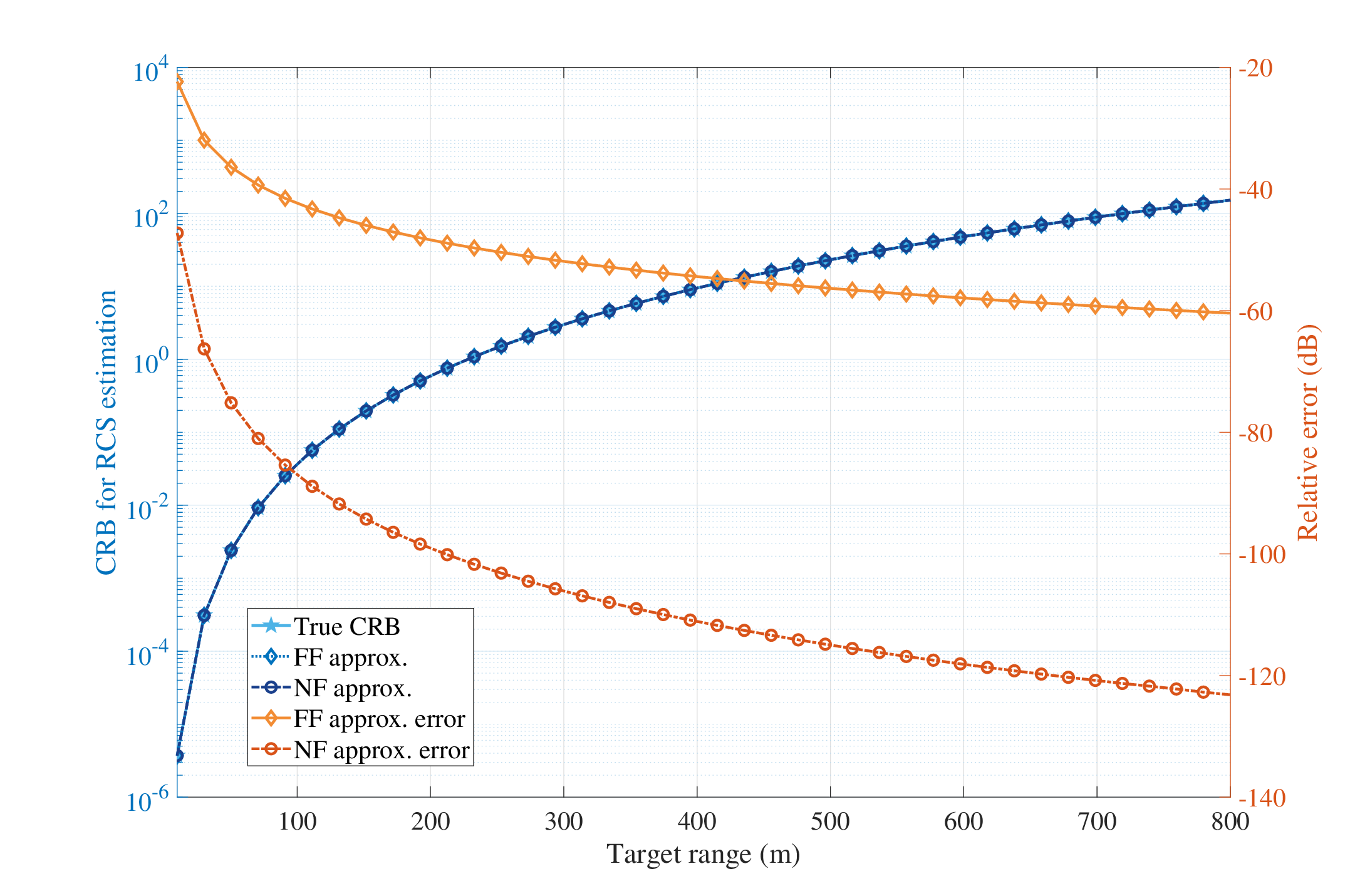}}
\vspace{-0.2cm}\caption{CRB for RCS versus target range, ($N_t=256$)
\label{RCS_range}}
\end{figure}

Fig.~\ref{RCS_angle} depicts the CRB for estimating the target RCS versus the target angle with $N_t=256$ and $r_{\mathrm{Tx},\mathrm{Rx}}=100~m$. The true CRB exhibits a clear angular dependence. Specifically, the estimation accuracy is best around the broadside and gradually degrades as the target moves toward the front of the array. This behavior is consistent with the aperture geometry of a centered ULA. Around broadside, the two-way phase profile across the array is the most ``balanced'' and provides stronger sensitivity to a global complex scaling (i.e., $\alpha_q$). Notably, the FF approximation yields an almost angle-invariant CRB, remaining essentially constant across the considered angles, since the plane-wave model neglects near-field curvature effects.
Fig.~\ref{RCS_angle} also compares the NF approximation and the FF approximation against the true CRB and reports the corresponding relative errors on the right axis. The NF approximation closely matches the true CRB across the full angular range, with a consistently small relative error. In contrast, the FF approximation introduces noticeable deviations, particularly in angular regions where near-field wavefront curvature and element-wise amplitude variations are non-negligible. The error curves further confirm that the FF approximation is only reliable when the target is sufficiently far and the curvature across the synthesized aperture becomes negligible, whereas the proposed NF approximation provides a uniformly accurate surrogate for narrow-band near-field performance evaluation.
\begin{figure}[t]
\centering{\includegraphics[width=1\columnwidth]{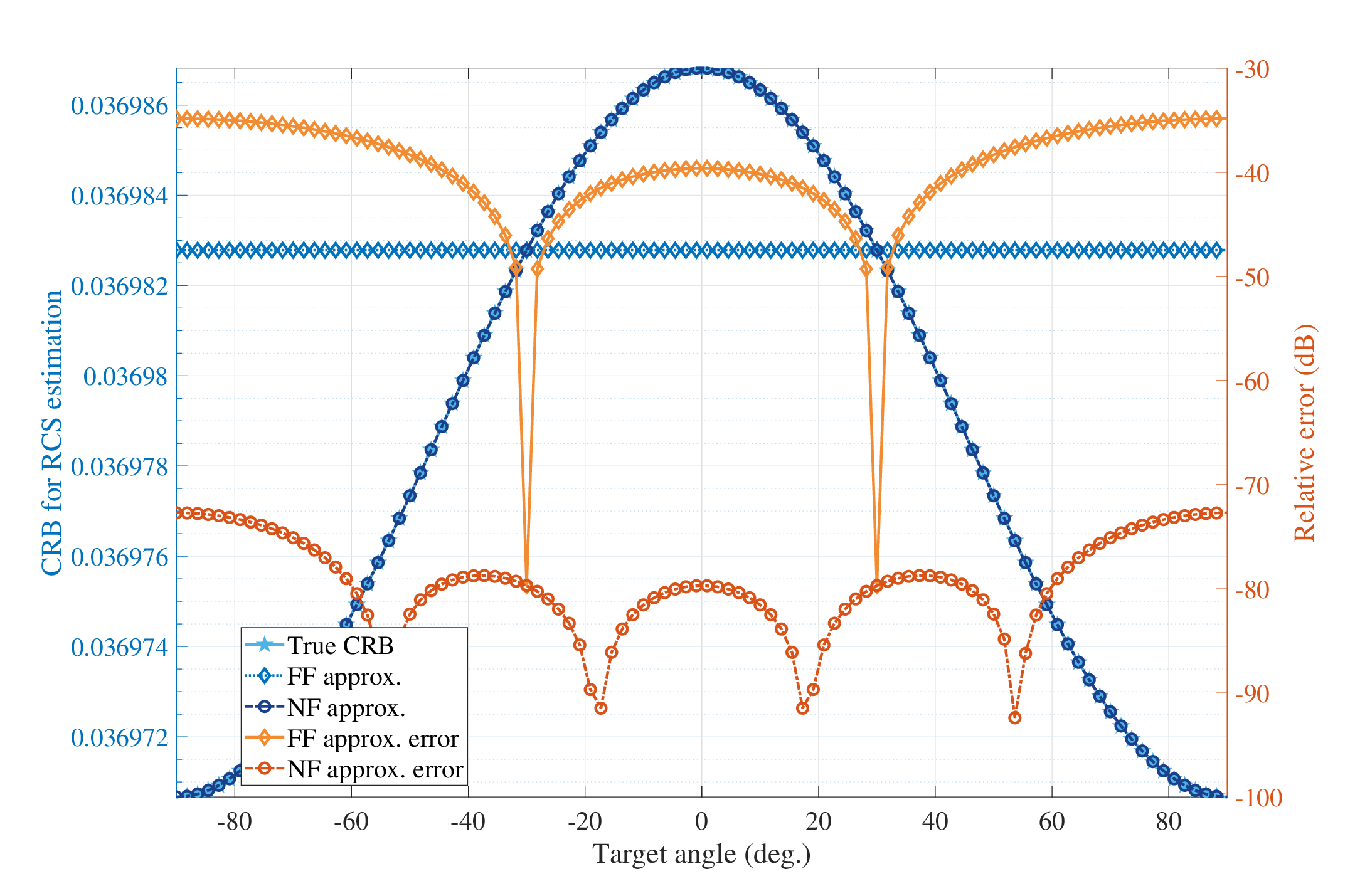}}
\caption{CRB for RCS versus target angle, ($N_t=256$)
\label{RCS_angle}}
\end{figure}

Fig.~5 shows the CRB for RCS estimation versus the number of antennas. It is seen that as $N_t$ increases, the true CRB decreases monotonically, confirming that a larger aperture provides higher array gain and more measurements, thereby increasing the Fisher information for the reflectivity parameter. Meanwhile,
the approximation accuracy is also reported via the relative error (right axis). Both NF approximation and FF approximation remain tightly aligned with the true CRB over the entire range of antenna number. 
\begin{figure}[t]
\centering{\includegraphics[width=1\columnwidth]{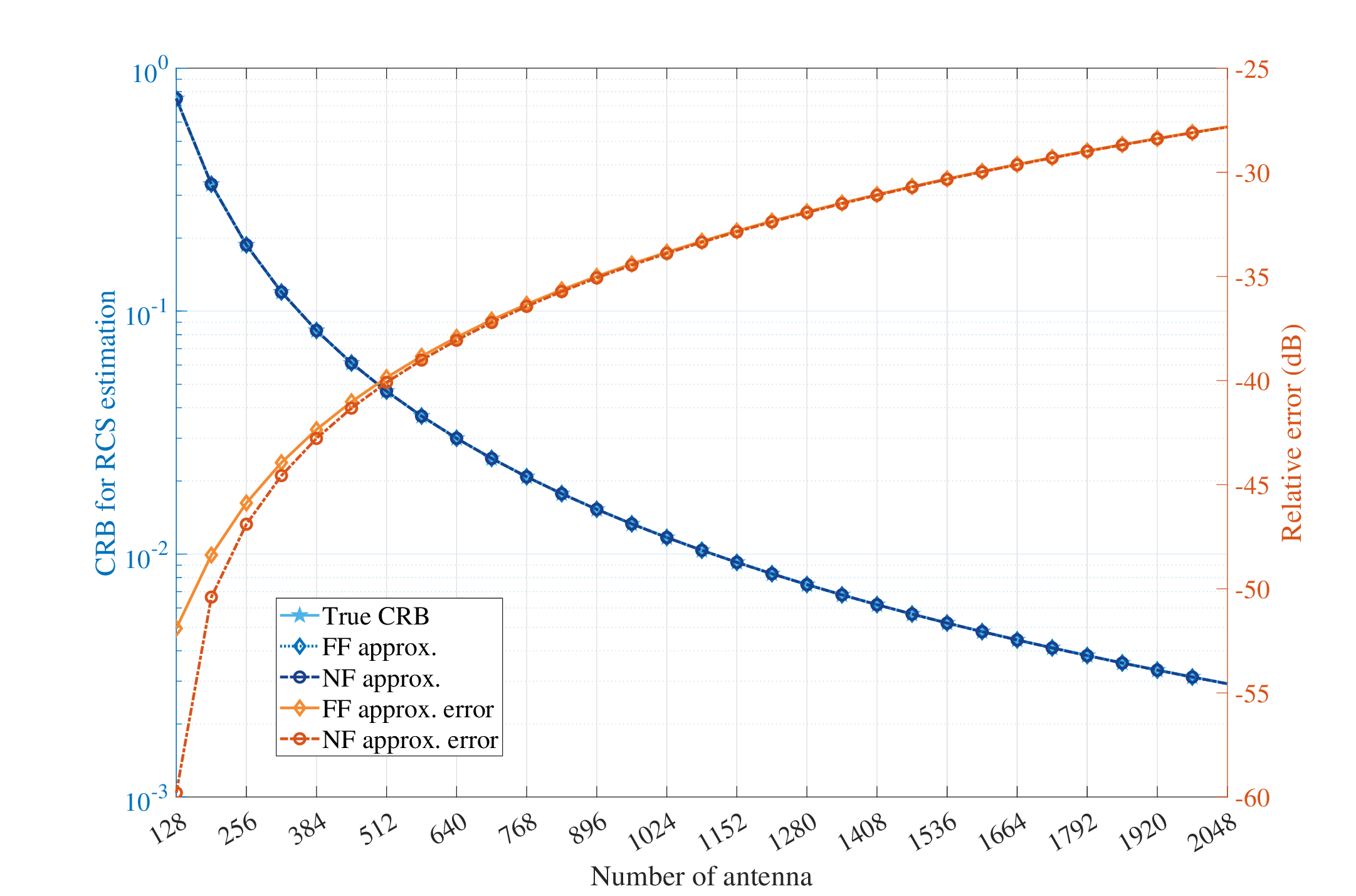}}
\caption{CRB for RCS versus antenna number
\label{RCS_antenna}}
\end{figure}
\vspace{-0.2cm}
\subsection{CRB for Velocity with Single Target and Monostatic Setting}
\vspace{-0.1cm}
We next investigate the velocity estimation limits in the single-target monostatic setting. It is seen from Fig.\ref{vel_range} that the true CRBs for both $v_x$ and $v_y$ increase monotonically with range, indicating that velocity estimation becomes progressively more challenging as the target moves away from the array. This behavior follows from the reduced received signal strength and the weakened Doppler sensitivity at larger distances, which jointly decrease the Fisher information associated with the slow-time phase evolution.
Fig.~\ref{vel_range} also compares the NF approximation and the FF approximation against the true CRB, with the relative errors reported on the right axis. The proposed NF approximation closely matches the true CRB for both velocity components over the entire range interval, yielding consistently small errors.  In contrast, the FF approximation exhibits noticeable discrepancies in the near-field regime, e.g. from short to moderate ranges, while its accuracy improves as the target range increases and the propagation approaches the plane-wave condition. Overall, these results confirm that near-field wavefront curvature and element-wise range dependence must be accounted for to correctly characterize velocity estimation limits with electrically large apertures, whereas far-field simplifications are reliable only at sufficiently large ranges

\begin{figure}[t]
\centering{\includegraphics[width=1\columnwidth]{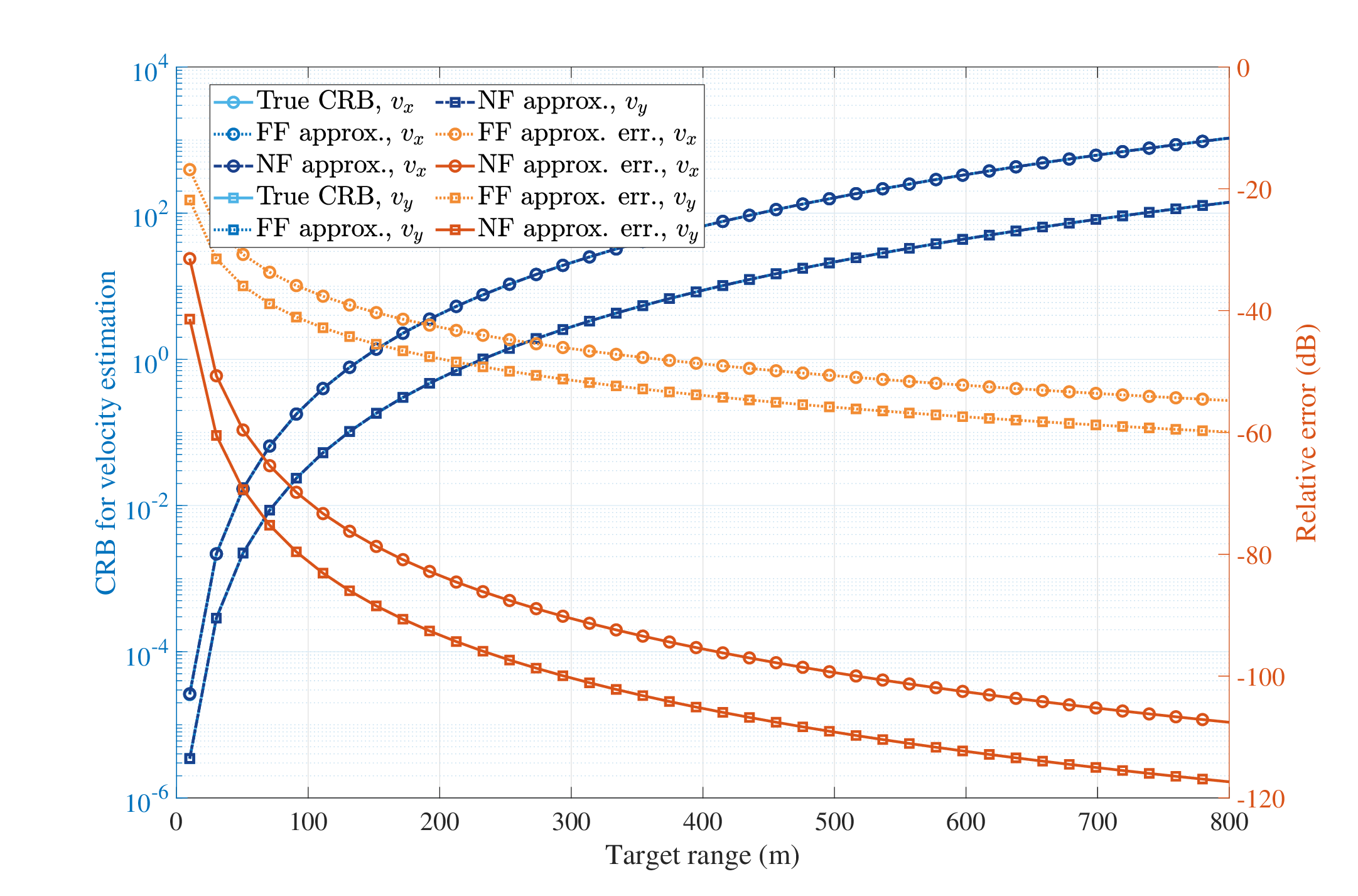}}
\caption{CRB for velocity versus target range, ($N_t=256$)
\label{vel_range}}
\end{figure}

Fig.~\ref{vel_angle} depicts the CRBs for estimating the velocity components versus the target angle with $N_t=256$. The true bounds exhibit a pronounced angular dependence and strong anisotropy between $v_x$ and $v_y$. In particular, the CRB for $v_x$ becomes sharply worse around broadside ($\theta\approx 0^\circ$), whereas the CRB for $v_y$ is comparatively smaller near broadside and degrades toward larger $|\theta|$. This trend is consistent with the underlying geometry. Specifically, the Doppler sensitivity of each Cartesian velocity component is governed by its projection onto the line-of-sight directions from the array elements to the target, which varies with $\theta$ and leads to angle-dependent observability. Furthermore, Fig.~\ref{vel_angle} also compares the NF and FF approximations with the true CRB and reports the relative errors on the right axis. The proposed NF approximation remains tightly aligned with the true CRB for both $v_x$ and $v_y$ over the entire angular range, with consistently small relative errors. 

\begin{figure}[t]
\centering{\includegraphics[width=1\columnwidth]{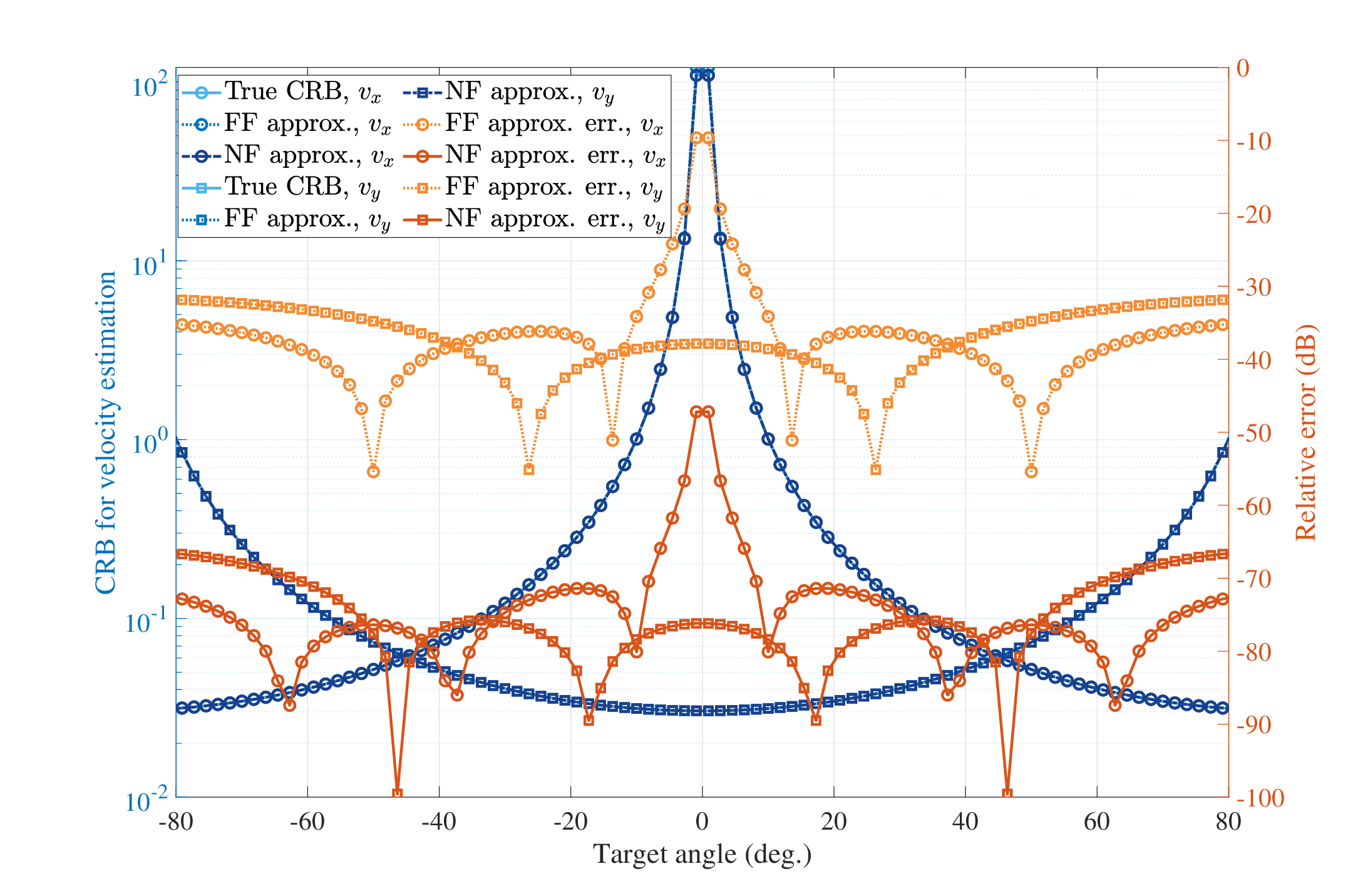}}
\caption{CRB for velocity versus target angle, ($N_t=256$)
\label{vel_angle}}
\end{figure}
Fig.~\ref{vel_antenna} shows the velocity CRBs versus the number of antennas. The true CRBs for both $v_x$ and $v_y$ decrease monotonically as the array size grows, confirming that larger apertures provide stronger signal gain and richer spatial diversity, which translate into higher Fisher information for the slow-time Doppler-induced phase evolution. Moreover, the improvement is more pronounced for the component that is better aligned with the line-of-sight sensitivity under the considered geometry, leading to different decay rates for $v_x$ and $v_y$. The right axis reports the relative errors of the NF and FF approximations. The proposed NF approximation closely tracks the true CRB over the entire antenna-number range, resulting in consistently small errors. In contrast, the FF approximation becomes increasingly inaccurate as $N_t$ increases. Specifically, its relative error grows with array size because near-field curvature and element-dependent range variations become more significant for electrically large apertures, violating the plane-wave assumption. This figure again highlights that far-field-based simplifications may substantially misestimate velocity limits for large arrays, whereas the NF approximation remains reliable in narrow-band near-field Doppler sensing.

\begin{figure}[t]
\centering{\includegraphics[width=1\columnwidth]{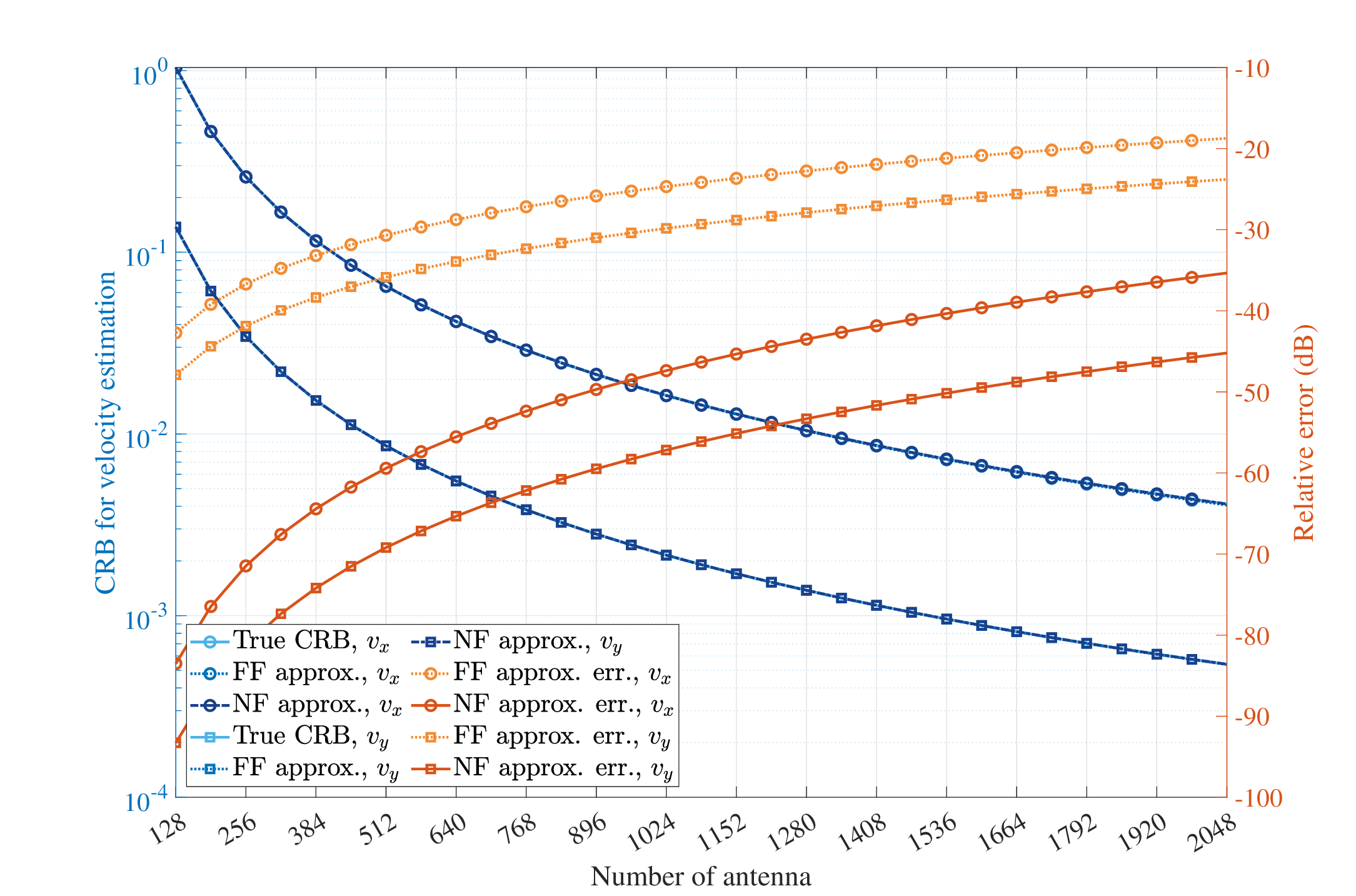}}
\caption{CRB for velocity versus antenna number. 
\label{vel_antenna}}
\end{figure}

\vspace{-0.2cm}
\subsection{CRB for Location with Single Target and Monostatic Setting}
\vspace{-0.1cm}
Then, we report the CRBs for 2D location estimation. We first consider the static target and set ${\bf v}=[0,0]$. 
Fig.~\ref{loc_range} reports the location CRBs versus target range. As the target moves farther away, the true CRBs for both $x$ and $y$ increase monotonically, indicating deterioration in the localization accuracy in larger ranges. This trend is expected since the received signal strength decreases with distance, and the wavefront curvature across the aperture becomes weaker, both of which reduce the sensitivity of the measurements to small perturbations in the target position.
Fig.~\ref{loc_range} further compares the NF and FF approximations with the true CRB, and plots the corresponding relative errors on the right axis. The NF approximation remains tightly matched to the true CRB across the entire range sweep, yielding consistently small errors for both $x$ and $y$. In contrast, the FF approximation exhibits noticeable discrepancies in the near-field region, where the plane-wave assumption neglects the range-dependent phase curvature and element-wise amplitude variations that carry localization information. As the range increases and the target transitions toward the far-field regime, the FF approximation becomes more accurate and its relative error decreases. Overall, Fig.~\ref{loc_range} confirms that accurate near-field modeling is essential for correctly characterizing narrow-band localization limits with large apertures, whereas far-field simplifications are only reliable at sufficiently large ranges.

 \begin{figure}[t]
\centering{\includegraphics[width=1\columnwidth]{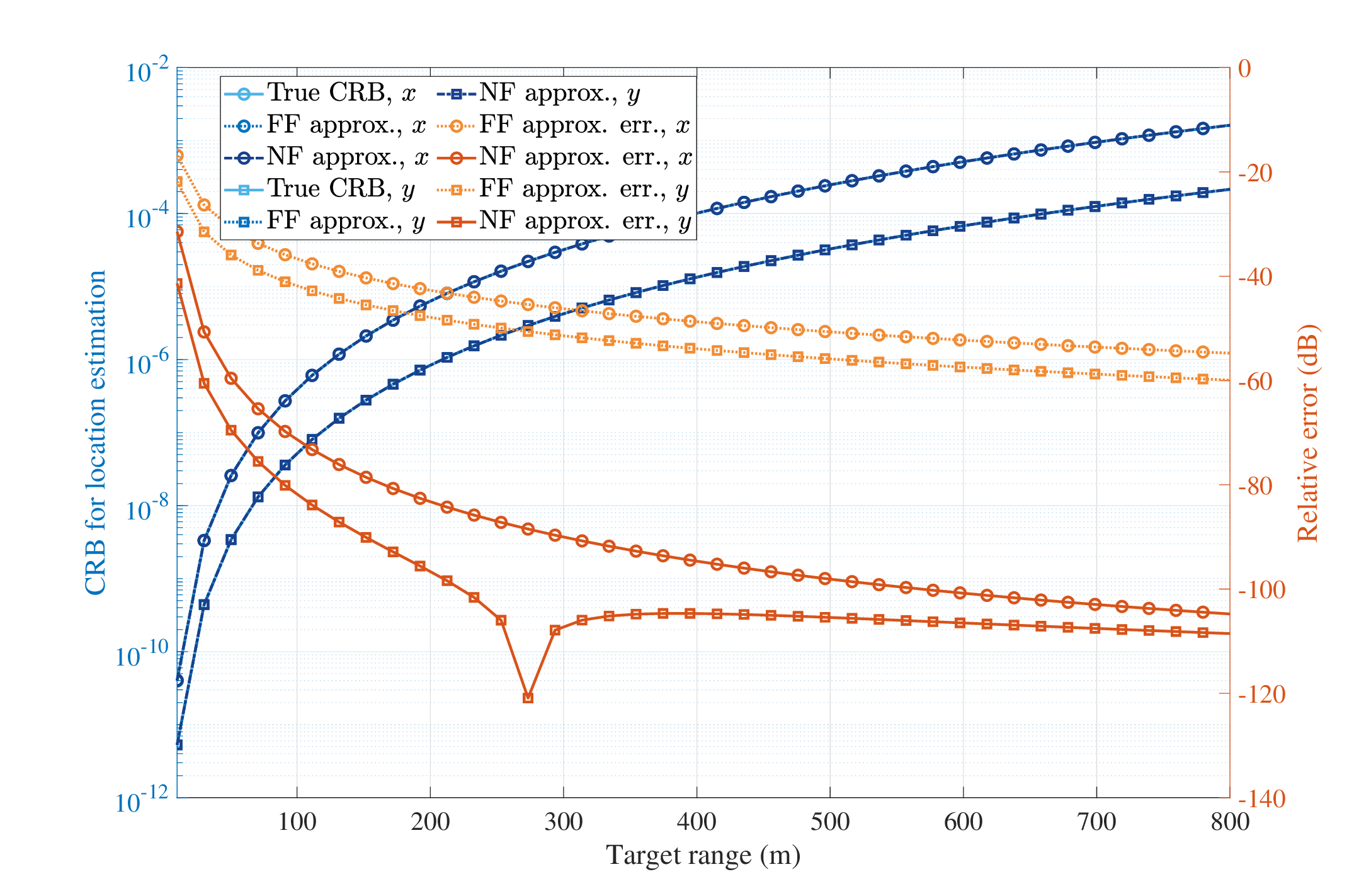}}
\caption{CRB for location versus target range, ($N_t=256$)
\label{loc_range}}
\end{figure}

Fig.~\ref{loc_angle} shows the location CRBs versus the target angle with $N_t=256$. The true CRBs exhibit a strong angular dependence and clear anisotropy between the two Cartesian coordinates. Specifically, the CRB for the cross-range component (associated with the lateral displacement relative to the array broadside) deteriorates sharply around broadside ($\theta\approx 0^\circ$), forming a pronounced peak. This behavior is caused by the intrinsic symmetry of a centered linear aperture, which reduces the sensitivity of the received phase to small lateral perturbations when the target lies close to the array normal. As the target moves away from broadside, the spatial phase gradients become stronger and the localization bounds improve accordingly. In contrast, the other coordinate maintains a comparatively flatter profile, reflecting its different geometric projection onto the element-to-target directions. The relative-error curves further compare the NF and FF approximations. The proposed NF approximation closely tracks the true CRB over the entire angular range, leading to consistently small errors. By contrast, the FF approximation produces an almost angle-invariant bound and exhibits noticeable mismatch, particularly near broadside where near-field curvature and parameter coupling dominate. These results verify that far-field plane-wave modeling can severely mischaracterize the angular dependence of near-field localization limits, while the NF approximation provides an accurate surrogate for narrow-band near-field analysis.

\begin{figure}[t]
\centering{\includegraphics[width=1\columnwidth]{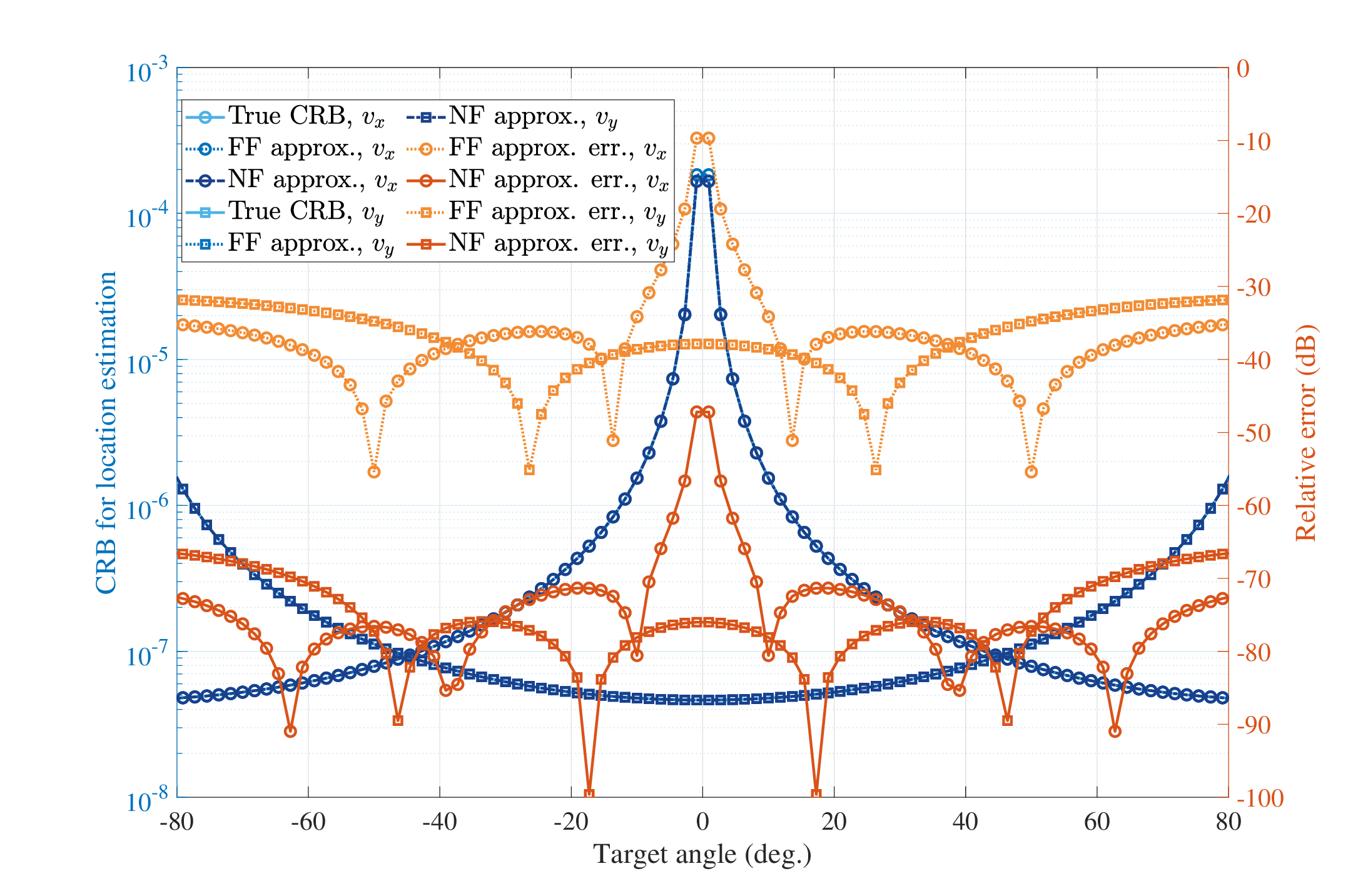}}
\caption{CRB for location versus target angle, ($N_t=256$)
\label{loc_angle}}
\end{figure}

Now, we consider the moving target with ${\bf v}= [1, 4]$ m/s. 
Fig.~\ref{loc_antenna} evaluates the location CRBs versus the number of antennas for a moving target in the monostatic configuration. Compared with the static-target case, motion introduces an additional slow-time phase evolution and, more importantly, non-negligible coupling between the location and velocity parameters in the Fisher information matrix. Nevertheless, the overall trend remains clear: the true CRBs for both $x$ and $y$ decrease monotonically as $N_t$ increases. This confirms that enlarging the aperture improves localization accuracy by enhancing both array gain and spatial phase diversity, even when the target is moving.
The right axis reports the relative errors of the NF and FF approximations. The NF approximation remains closely aligned with the true CRB across the full antenna-number range, indicating that the proposed near-field simplification accurately preserves the motion-induced coupling effects in the narrow-band regime. In contrast, the FF approximation becomes increasingly unreliable as $N_t$ grows. Its relative error increases with array size because the plane-wave assumption neglects near-field curvature and element-dependent range variations, and it also fails to capture the coupling between motion and geometry that becomes more pronounced for electrically large apertures. Overall, Fig.~\ref{loc_antenna} demonstrates that accurate near-field modeling is essential for characterizing localization limits for moving targets with large monostatic apertures.

\begin{figure}[t]
\centering{\includegraphics[width=1\columnwidth]{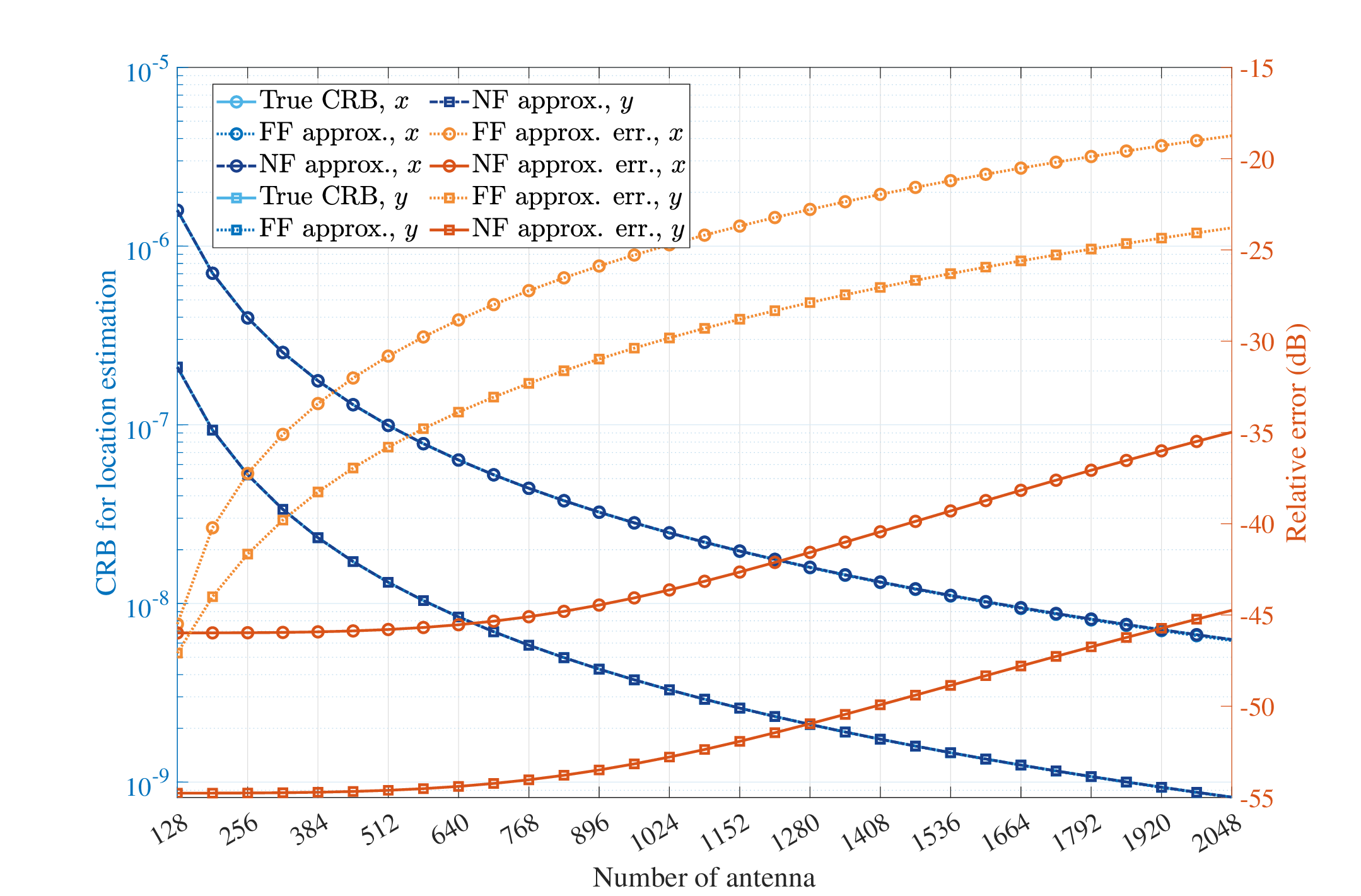}}
\caption{CRB for location versus antenna number 
\label{loc_antenna}}
\end{figure}


\vspace{-0.2cm}
\subsection{CRB for Velocity with Multiple Targets and Bistatic Setting}
\vspace{-0.1cm}

Fig.~\ref{vel_antenna2} considers a bistatic setup with three targets, where the Tx and Rx centroids are separated (Tx at $x_{t_0}=-2$, Rx at $x_{r_0}=2$). The targets are located at $(20^\circ,100~\mathrm{m})$ with ${\bm v}=[1,4]$~m/s, $(-45^\circ,150~\mathrm{m})$ with ${\bm v}=[4,3]$~m/s, and $(-5^\circ,50~\mathrm{m})$ with ${\bm v}=[10,6]$~m/s. As $N_t$ increases, the true CRBs for both $v_x$ and $v_y$ decrease monotonically, showing that larger bistatic apertures improve velocity estimation through higher array gain and spatial diversity, despite additional multi-target coupling.
The NF approximation closely matches the true CRB over the entire antenna range with consistently small relative errors, whereas the FF approximation exhibits larger deviations and becomes less reliable for large arrays, since the plane-wave model cannot capture the near-field curvature and element-dependent bistatic geometry even in multiple and bistatic setup.

\begin{figure}[t]
\centering{\includegraphics[width=1\columnwidth]{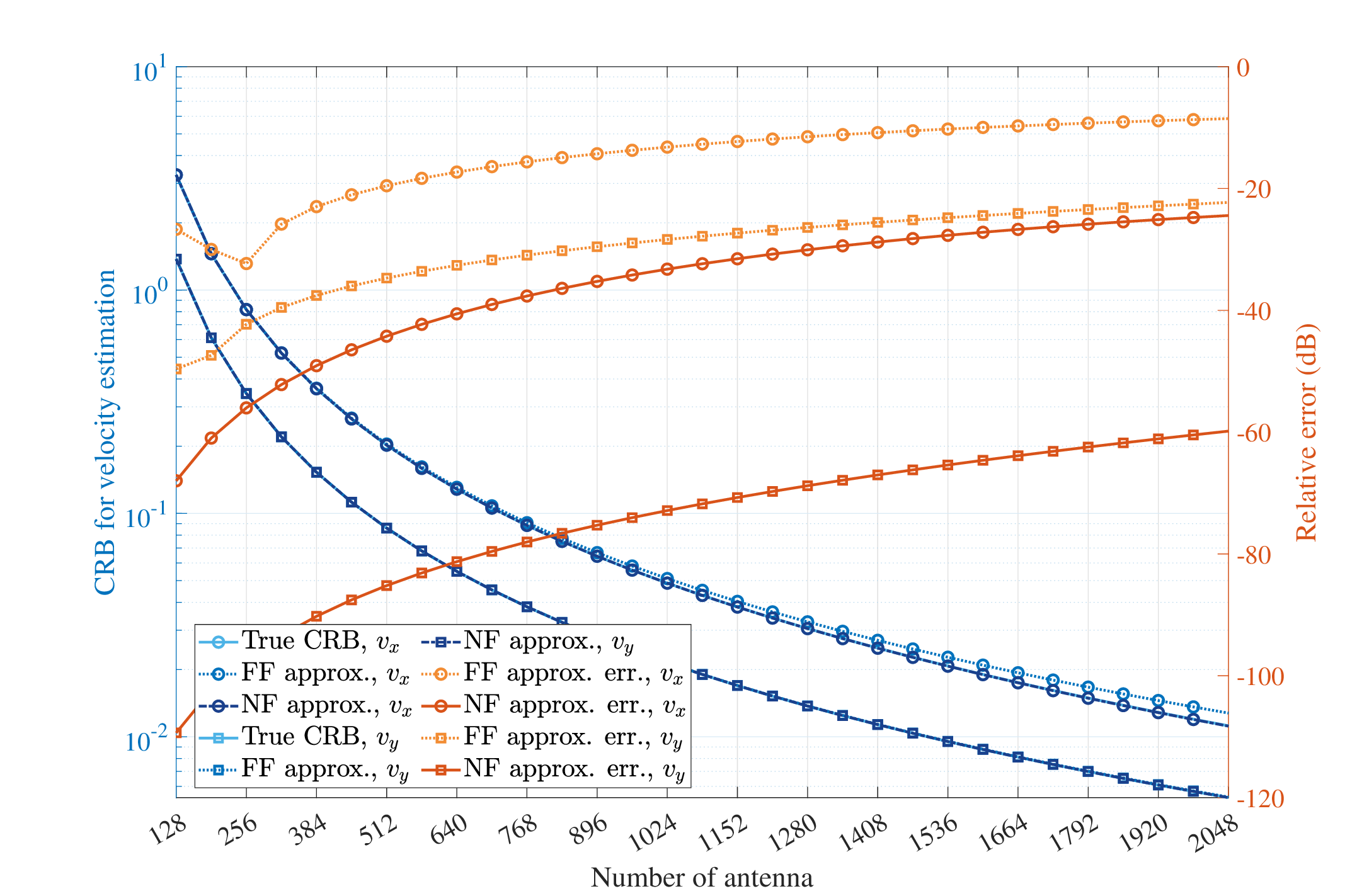}}
\caption{CRB for velocity versus antenna number. 
\label{vel_antenna2}}
\end{figure}

\section{Conclusions}
\label{sec:summ}

This paper (Part~I) investigated fundamental limits for narrow-band near-field sensing with electrically large apertures under a unified monostatic/bistatic geometric model. By explicitly accounting for wavefront curvature and element-wise range dependence, we derived conditional CRBs for the key parameters, including target reflectivity (RCS), velocity, and location, and highlighted the near-field coupling effects that vanish under far-field plane-wave assumptions. To enable efficient evaluation and provide design insight, we further developed a tractable near-field approximation and assessed its accuracy against the true CRB.

Numerical experiments validated the analysis and revealed several consistent trends: (i) all bounds degrade with increasing range and improve with larger aperture size; (ii) the CRBs exhibit strong angular dependence with pronounced broadside behaviors due to array symmetry; and (iii) the proposed near-field approximation remains tightly matched to the true CRB over wide operating regimes, whereas far-field approximations can incur substantial errors in the radiative near-field, especially for electrically large arrays and moving targets. These results provide practical guidelines for narrow-band synthetic-aperture/large-array sensing design and establish a rigorous baseline for Part~II \cite{wei2025PartII}, where we extend the framework to wideband systems and quantify the additional benefits brought by frequency diversity.

\appendices
\section{Proof of Theorem 1}
\label{app:proof_thm1}
First, for the far-field target, the distances between the target and each element are approximately the same.  Then, we have $g_{\mathrm {Tx}}^{\mathrm{FF}}\approx \frac{N_t}{{r_{\mathrm{Tx}}^{\mathrm{FF}}}^2}$ and $g_{\mathrm {Rx}}^{\mathrm{FF}}\approx \frac{N_r}{{r_{\mathrm{Rx}}^{\mathrm{FF}}}^2}$.
 The approximations in \eqref{CRB_RCS.1.1} hold. 

Let the array be centered at the origin so that the $n_t$-th Tx element,
$n_t=1,\ldots,N_t$, is located at
\begin{equation}
  x_{n_t} = \Bigl(n_t - \frac{N_t+1}{2}\Bigr)d_t , \qquad
  y_{n_t}=0 ,
\end{equation}
and therefore
\begin{equation} \label{gem_sum}
  \sum_{n_t=1}^{N_t} x_{n_t} = 0, \qquad
  \sum_{n_t=1}^{N_t} x_{n_t}^2 = \frac{N_t(N_t^{2}-1)d_t^{2}}{12}.
\end{equation}
Consider a single target located at polar coordinates
$(r_{\mathrm{Tx}},\theta_{\mathrm{Tx},q})$ with respect to the Tx array
centroid, i.e.,
\begin{equation}
  x_q = r_{\mathrm{Tx}}\sin\theta_{\mathrm{Tx},q}, \qquad
  y_q = r_{\mathrm{Tx}}\cos\theta_{\mathrm{Tx},q}.
\end{equation}
The distance between the target and the $n_t$-th Tx element is
\begin{align}
  r_{n_t}
  &= \bigl\|\mathbb{I}_q - \mathbb{I}_{n_t}\bigr\|
   = \sqrt{(x_q - x_{n_t})^{2} + y_q^{2}} \nonumber\\
  &= r_{\mathrm{Tx}}
     \sqrt{1
           - \frac{2x_{n_t}\sin\theta_{\mathrm{Tx},q}}{r_{\mathrm{Tx}}}
           + \frac{x_{n_t}^{2}}{r_{\mathrm{Tx}}^{2}} } .
  \label{eq:rnt_exact}
\end{align}
In the Fresnel region, the array aperture is much smaller than the range,
i.e., $D_t/2 \ll r_{\mathrm{Tx}}$, where
$D_t=(N_t-1)d_t$ is the Tx aperture. Hence, using the second-order
Taylor expansion of $\sqrt{1+\varepsilon}$ for
\[
  \varepsilon
  = - \frac{2x_{n_t}\sin\theta_{\mathrm{Tx},q}}{r_{\mathrm{Tx}}}
    + \frac{x_{n_t}^{2}}{r_{\mathrm{Tx}}^{2}},
\]
we obtain
\begin{align}
  r_{n_t}
  &\approx
  r_{\mathrm{Tx}}
   \left(1 + \frac{\varepsilon}{2} - \frac{\varepsilon^{2}}{8}\right) \nonumber\\
  &\approx
  r_{\mathrm{Tx}}
  - x_{n_t}\sin\theta_{\mathrm{Tx},q}
  + \frac{x_{n_t}^{2}\cos^{2}\theta_{\mathrm{Tx},q}}{2r_{\mathrm{Tx}}}.
  \label{eq:rnt_taylor}
\end{align}
To derive the array gain, we need $1/r_{n_t}^{2}$.  Writing
$r_{n_t}=r_{\mathrm{Tx}}(1+\delta_{n_t})$ with
\begin{equation}
  \delta_{n_t}
  = - \frac{x_{n_t}\sin\theta_{\mathrm{Tx},q}}{r_{\mathrm{Tx}}}
    + \frac{x_{n_t}^{2}\cos^{2}\theta_{\mathrm{Tx},q}}{2r_{\mathrm{Tx}}^{2}},
\end{equation}
and using the second-order expansion $(1+\delta)^{-2}\approx
1-2\delta+3\delta^{2}$ for $|\delta|\ll1$, we get
\begin{align}
 & \frac{1}{r_{n_t}^{2}}
  = \frac{1}{r_{\mathrm{Tx}}^{2}}(1+\delta_{n_t})^{-2} \nonumber\\
  &\approx \frac{1}{r_{\mathrm{Tx}}^{2}}
  \Biggl[
     1
     + \frac{2x_{n_t}\sin\theta_{\mathrm{Tx},q}}{r_{\mathrm{Tx}}}
     + \frac{x_{n_t}^{2}}{r_{\mathrm{Tx}}^{2}}
       \bigl(4\sin^{2}\theta_{\mathrm{Tx},q}-1\bigr)
  \Biggr].
  \label{eq:inv_rnt2_expand}
\end{align}
The total Tx gain is then
\begin{align}
  &g_{\mathrm{Tx}}^{\mathrm{NF}}
  = \sum_{n_t=1}^{N_t} \frac{1}{r_{n_t}^{2}} \label{eq:gTx_sum_def} \\
  &\approx \frac{1}{r_{\mathrm{Tx}}^{2}}
    \sum_{n_t=1}^{N_t}
    \Biggl[
       1
       + \frac{2x_{n_t}\sin\theta_{\mathrm{Tx},q}}{r_{\mathrm{Tx}}}
       + \frac{x_{n_t}^{2}}{r_{\mathrm{Tx}}^{2}}
         \bigl(4\sin^{2}\theta_{\mathrm{Tx},q}-1\bigr)
    \Biggr] \nonumber\\
  &= \frac{N_t}{r_{\mathrm{Tx}}^{2}}
     + \frac{2\sin\theta_{\mathrm{Tx},q}}{r_{\mathrm{Tx}}^{3}}
       \sum_{n_t=1}^{N_t} x_{n_t}
     + \frac{4\sin^{2}\theta_{\mathrm{Tx},q}-1}{r_{\mathrm{Tx}}^{4}}
       \sum_{n_t=1}^{N_t} x_{n_t}^{2}. \label{eq:gTx_sum_expand}
\end{align}
Recalling \eqref{gem_sum}, we finally obtain
\begin{equation} \label{gain_NF_Tx}
  g_{\mathrm{Tx}}^{\mathrm{NF}}
  \approx
  \frac{N_t}{r_{\mathrm{Tx}}^{2}}
  + \frac{N_t(N_t^{2}-1)d_t^{2}}{12r_{\mathrm{Tx}}^{4}}
    \bigl(4\sin^{2}\theta_{\mathrm{Tx},q}-1\bigr).
\end{equation}
It is observed that the near-field channel gain of Tx in \eqref{gain_NF_Tx} is not related to the relative position of each antenna element. Hence, for bistatic case, a similar derivation for the Rx array with $N_r$ elements and spacing
$d_r$ yields
\begin{equation} \label{gain_NF_Rx}
  g_{\mathrm{Rx}}^{\mathrm{NF}}
  \approx
  \frac{N_r}{r_{\mathrm{Rx}}^{2}}
  + \frac{N_r(N_r^{2}-1)d_r^{2}}{12r_{\mathrm{Rx}}^{4}}
    \bigl(4\sin^{2}\theta_{\mathrm{Rx},q}-1\bigr),
\end{equation}
where $r_{\mathrm{Rx}}$ and $\theta_{\mathrm{Rx},q}$ denote the target
range and angle with respect to the Rx array centroid.
Substituting \eqref{gain_NF_Tx} and \eqref{gain_NF_Rx} into \eqref{gain_Tx} and \eqref{gain_Rx}, we get near-field channel gain
\begin{subequations} \label{gain_NF_Tx_Rx}
 \begin{flalign} 
   G_{\mathrm{Tx}}^{\mathrm{NF}} & \!\approx\! \frac{\lambda_c^2N_t}{16\pi^2r_{\mathrm{Tx}}^{2}}(1\!+\!\frac{(N_t^{2}-1)d_t^{2}}{12r_{\mathrm{Tx}}^{2}}(4\sin^{2}\theta_{\mathrm{Tx},q}\!-\!1)), \label{gain_NF_Tx_final}\\ 
   G_{\mathrm{Rx}}^{\mathrm{NF}} & \!\approx\! \frac{\lambda_c^2N_r}{16\pi^2r_{\mathrm{Rx}}^{2}}(1\!+\!\frac{(N_r^{2}-1)d_r^{2}}{12r_{\mathrm{Rx}}^{2}}(4\sin^{2}\theta_{\mathrm{Rx},q}\!-\!1)). \label{gain_NF_Rx_final}
 \end{flalign}   
\end{subequations}
Finally, substituting \eqref{gain_NF_Tx} and \eqref{gain_NF_Rx} into 
\eqref{CRB_RCS}, we can directly obtain \eqref{CRB_RCS.1.2}. This completes the proof.

\section{Proof of Theorem 2}
Analogously, we introduce the partial channel gains associated with the velocity components as
\begin{small}
\begin{subequations} \label{gain_Tx_Rx_vx_vy}
 \begin{align} 
\hspace*{-0.5em}   \acute{G}_{\mathrm{Tx}}(m) &\!=\! \|\acute{\bf a}_{\mathrm{T}}(m,q)\|_2^2 = \frac{ m^2 T_{\mathrm{sym}}^2}
{4}\sum_{n_t=1}^{N_t}\frac{(x_q - x_{n_t})^2}{r_{n_t}^4}, \label{gain_Tx_vx}\\ 
\hspace*{-0.5em}   \acute{G}_{\mathrm{Rx}}(m) &\!=\! \|\acute{\bf a}_{\mathrm{R}}(m,q)\|_2^2 = \frac{ m^2 T_{\mathrm{sym}}^2}
{4}\sum_{n_r=1}^{N_r}\frac{(x_q - x_{n_r})^2}{r_{n_r}^4}, \label{gain_Rx_vx}\\
\hspace*{-0.5em}   \grave{G}_{\mathrm{Tx}}(m) &\!=\! \|\grave{\bf a}_{\mathrm{T}}(m,q)\|_2^2 = \frac{ m^2 T_{\mathrm{sym}}^2}
{4}\sum_{n_t=1}^{N_t}\frac{(y_q - y_{n_t})^2}{r_{n_t}^4}, \label{gain_Tx_vy}\\ 
\hspace*{-0.5em}   \grave{G}_{\mathrm{Rx}}(m) &\!=\! \|\grave{\bf a}_{\mathrm{R}}(m,q)\|_2^2 = \frac{ m^2 T_{\mathrm{sym}}^2}
{4}\sum_{n_r=1}^{N_r}\frac{(y_q - y_{n_r})^2}{r_{n_r}^4}. \label{gain_Rx_vy}
 \end{align}   
\end{subequations}
\end{small}
\noindent\hspace{-1em} Again, for the far-field target, we have $\|\mathbb I_q - \mathbb I_{n_t}\|\approx r_{\mathrm{Tx}}^{\mathrm{FF}}$, $x_q \approx r_{\mathrm{Tx}}^{\mathrm{FF}}\sin\theta_{\mathrm{Tx},q}$. Then, we can far-field partial gain approximation as
\begin{small}
\begin{subequations} \label{FF_gain_Tx_Rx_vx_vy}
 \begin{align} 
   \acute{G}_{\mathrm{Tx}}^{\mathrm{FF}}(m) & \approx \frac{N_t m^2 T_{\mathrm{sym}}^2}{4 {r_{\mathrm{Tx}}^{\mathrm{FF}}}^2}
\left(\sin^2\theta_{\mathrm{Tx},q} \!+\! \frac{(N_t^2-1)d_t^2}{12{r_{\mathrm{Tx}}^{\mathrm{FF}}}^2}
\right), \label{FF_approx_gain_Tx_vx}\\ 
   \acute{G}_{\mathrm{Rx}}^{\mathrm{FF}}(m) & \approx \frac{N_r m^2 T_{\mathrm{sym}}^2}{4 {r_{\mathrm{Rx}}^{\mathrm{FF}}}^2}
\left(\sin^2\theta_{\mathrm{Rx},q} \!+\! \frac{(N_r^2-1)d_r^2}{12{r_{\mathrm{Rx}}^{\mathrm{FF}}}^2}
\right), \label{FF_approx_gain_Rx_vx}\\
   \grave{G}_{\mathrm{Tx}}^{\mathrm{FF}}(m) & \approx \frac{N_t m^2 T_{\mathrm{sym}}^2}{4{r_{\mathrm{Tx}}^{\mathrm{FF}}}^2}
\cos^2\theta_{\mathrm{Tx},q}, \label{FF_approx_gain_Tx_vy}\\ 
   \grave{G}_{\mathrm{Rx}}^{\mathrm{FF}}(m) & \approx \frac{N_r m^2 T_{\mathrm{sym}}^2}{4{r_{\mathrm{Rx}}^{\mathrm{FF}}}^2}
\cos^2\theta_{\mathrm{Rx},q}. \label{FF_approx_gain_Rx_vy}
 \end{align}   
\end{subequations}
\end{small}
The cross term $\Re\!\left\{ 
\acute{\mathbf{a}}_{\mathrm{R}}^{H}(m,q)\mathbf{a}_{\mathrm{R}}(m,q)
\mathbf{a}_{\mathrm{T}}^{H}(m,q)\acute{\mathbf{a}}_{\mathrm{T}}(m,q)
\right\}$ and $\Re\!\left\{ 
\grave{\mathbf{a}}_{\mathrm{R}}^{H}(m,q)\mathbf{a}_{\mathrm{R}}(m,q)
\mathbf{a}_{\mathrm{T}}^{H}(m,q)\grave{\mathbf{a}}_{\mathrm{T}}(m,q)
\right\}$ can be calculated by
\begin{small}
\begin{subequations}
\begin{align}
&\hspace*{-0.5em}\acute{G}_{\mathrm{cr.}}(m)  = 
\frac{\lambda_k^2 m^2 T_{\mathrm{sym}}^2}{64\pi^2}
\sum_{n_r=1}^{N_r} \sum_{n_t=1}^{N_t}
\frac{(x_q - x_{n_r})(x_q - x_{n_t})}{r_{n_r}^3 r_{n_t}^3},\\
& \hspace*{-0.5em}\grave{G}_{\mathrm{cr.}}(m) = 
\frac{\lambda_k^2 m^2 T_{\mathrm{sym}}^2}{64\pi^2}
\sum_{n_r=1}^{N_r} \sum_{n_t=1}^{N_t}
\frac{y_q^2}{r_{n_r}^3 r_{n_t}^3}.
\end{align}
\end{subequations}
\end{small}
Its far-field approximation can be directly given by 
\begin{subequations} \label{cross_approx_FF}
\begin{align} 
&\acute{G}_{\mathrm{cr.}}^{\mathrm{FF}}(m) \approx \frac{N_t N_r \lambda_k^{2} m^{2} T_{\mathrm{sym}}^{2}}
{64 \pi^{2} r_{\mathrm{Tx}}^{2} r_{\mathrm{Rx}}^{2}}
\sin\theta_{\mathrm{Tx},q}\sin\theta_{\mathrm{Rx},q},\\
& \grave{G}_{\mathrm{cr.}}^{\mathrm{FF}}(m) \approx \frac{N_t N_r \lambda_k^{2} m^{2} T_{\mathrm{sym}}^{2}}
{64 \pi^{2} r_{\mathrm{Tx}}^{2} r_{\mathrm{Rx}}^{2}}
\cos\theta_{\mathrm{Tx},q}\cos\theta_{\mathrm{Rx},q}.
\end{align}
\end{subequations}
Substituting \eqref{FF_gain_Tx_Rx_vx_vy} and \eqref{cross_approx_FF} into \eqref{CRB_vel.q1_single} and \eqref{CRB_vel.q2_single}, we can directly get 
\eqref{CRB_vel_FF_final}.

For the near-field target, we have 
\begin{small}
\begin{subequations} \label{NF_gain_Tx_Rx_vx_vy}
 \begin{align} 
   \acute{G}_{\mathrm{Tx}}^{\mathrm{NF}}(m) & \approx 
\frac{N_t m^{2} T_{\mathrm{sym}}^{2}}{4{r_{\mathrm{Tx}}^{\mathrm{NF}}}^{2}}
B_{\mathrm{Tx}}^{(x)}, \label{NF_approx_gain_Tx_vx}\\ 
   \acute{G}_{\mathrm{Rx}}^{\mathrm{NF}}(m) & \approx \frac{N_r m^{2} T_{\mathrm{sym}}^{2}}{4{r_{\mathrm{Rx}}^{\mathrm{NF}}}^{2}}
B_{\mathrm{Rx}}^{(x)}, \label{NF_approx_gain_Rx_vx}\\
   \grave{G}_{\mathrm{Tx}}^{\mathrm{NF}}(m) & \approx \frac{N_t m^{2} T_{\mathrm{sym}}^{2}}{4{r_{\mathrm{Tx}}^{\mathrm{NF}}}^{2}}
B_{\mathrm{Tx}}^{(Y)}, \label{NF_approx_gain_Tx_vy}\\ 
   \grave{G}_{\mathrm{Rx}}^{\mathrm{NF}}(m) & \approx \frac{N_r m^{2} T_{\mathrm{sym}}^{2}}{4{r_{\mathrm{Rx}}^{\mathrm{NF}}}^{2}}
B_{\mathrm{Rx}}^{(y)}. \label{NF_approx_gain_Rx_vy}
 \end{align}   
\end{subequations}
\end{small}
\noindent\hspace{-0.5em}Similarly, the cross term approximation of near-field target can be given by 
\begin{small}
\begin{subequations} \label{cross_approx_NF}
\begin{align} 
&\acute{G}_{\mathrm{cr.}}^{\mathrm{NF}}(m) \!\approx\! \frac{N_t N_r \lambda_k^{2} m^{2} T_{\mathrm{sym}}^{2}}
{64 \pi^{2} r_{\mathrm{Tx}}^{2} r_{\mathrm{Rx}}^{2}}
\sin\theta_{\mathrm{Tx},q}\sin\theta_{\mathrm{Rx},q}\left(1\!+\!\Delta_{x,\mathrm{Tx}}^{\mathrm{NF}}\!+\!\Delta_{x,\mathrm{Rx}}^{\mathrm{NF}}
\right)
,\\
& \grave{G}_{\mathrm{cr.}}^{\mathrm{NF}}(m) \!\approx\! \frac{N_t N_r \lambda_k^{2} m^{2} T_{\mathrm{sym}}^{2}}
{64 \pi^{2} r_{\mathrm{Tx}}^{2} r_{\mathrm{Rx}}^{2}}
\cos\theta_{\mathrm{Tx},q}\cos\theta_{\mathrm{Rx},q}\left(1\!+\!\Delta_{y,\mathrm{Tx}}^{\mathrm{NF}}\!+\!\Delta_{y,\mathrm{Rx}}^{\mathrm{NF}}\right).
\end{align}
\end{subequations}
\end{small}
\noindent\hspace{-0.5em}Substituting \eqref{gain_NF_Tx_Rx}, \eqref{NF_gain_Tx_Rx_vx_vy} and \eqref{cross_approx_NF} into \eqref{CRB_loc.q1_single} and \eqref{CRB_loc.q2_single}, we can directly obtain \eqref{CRB_vel_NF_final}. This completes the proof.


\ifCLASSOPTIONcaptionsoff
  \newpage
\fi

\bibliographystyle{IEEEtran}
\bibliography{ref}

\ifCLASSOPTIONcaptionsoff
  \newpage
\fi

\end{document}